\documentclass[11pt]{amsart}
\usepackage[a4paper,margin=23mm,top=30mm]{geometry}
\usepackage{amsfonts, amsmath, amssymb, amsgen, amsthm, amscd}
\usepackage{newtxtext,newtxmath}
\usepackage[utf8]{inputenc} 

\usepackage{color}

\usepackage[all]{xy}
\usepackage{xcolor}
\usepackage{hyperref}
\usepackage{mathtools,slashed}
\usepackage{setspace}
\def\a{\alpha}

\def\d{\delta}

\def\s{\sigma}

\def\vp{\varphi}

\def\ot{\otimes}

\def\rt{\triangleright}
\def\lt{\triangleleft}

\def\Ad{\mathop{\rm Ad}\nolimits}
\def\ad{\mathop{\rm ad}\nolimits}

\usepackage{enumerate}

\newcommand{\G}[1]{\mathfrak{#1}}

\newcommand{\B}[1]{\mathbb{#1}}

\numberwithin{equation}{section}

\newtheorem{theorem}{Theorem}[section]
\newtheorem{proposition}[theorem]{Proposition}

\newtheorem{corollary}[theorem]{Corollary}
\theoremstyle{definition}

\newtheorem{remark}[theorem]{Remark}

\usepackage{tikz-cd}
\usepackage{setspace}

\par
\begin{document}
\title{Second Order Lagrangian Dynamics On Double Cross Product Groups}
\author{O\u{g}ul Esen}
\address{Department of Mathematics, Gebze Technical University,  41400 Gebze-Kocaeli, Turkey}
\email{oesen@gtu.edu.tr}

\author{Mahmut Kudeyt}
\address{Department of Mathematics, Gebze Technical University,  41400 Gebze-Kocaeli, Turkey}
\email{mahmutkudeyt@gmail.com}

\author{Serkan Sütlü}
\address{Department of Mathematics, I\c{s}ik University, 34980 \c{S}ile-\.{I}stanbul, Turkey}
\email{serkan.sutlu@isikun.edu.tr}
\date{}

\begin{abstract}
We  observe that the iterated tangent group of a Lie group may be realized as a double cross product of the 2nd order tangent group, with the Lie algebra of the base Lie group. Based on this observation, we derive the 2nd order Euler-Lagrange equations on the 2nd order tangent group from the 1st order Euler-Lagrange equations on the iterated tangent group. We also present in detail the 2nd order Lagrangian dynamics on the 2nd order tangent group of a double cross product group.

\bigskip

\noindent\footnotesize{\textbf{MSC 2010:} 22E70, 70H50.} \newline
\noindent\footnotesize{\textbf{Key words:} Euler-Poincar\'{e} equation, second order tangent bundles, matched pairs of Lie groups.}
\end{abstract}

\maketitle

\section{Introduction}

\onehalfspacing

The symmetry of a dynamical equation is defined as the invariance of the system under a (Lie) group action, \cite{de2011methods,  holm2008geometric, libermann2012symplectic, Ol93}. On the other hand, the reduction of Lagrangian systems under symmetries is one of the main interests of  geometric mechanics, \cite{n2001lagrangian}. More precisely, the reduction of the Euler-Lagrange equations, on the tangent bundle $TQ$ of the configuration space $Q$, reduces under the symmetry group (say $G$) action, to the Lagrange-Poincar\'{e} equations over the space $TQ/G$ of orbits, \cite{marsden1993lagrangian}. 

The configuration spaces of many physical systems such as the rigid bodies, or fluid and plasma theories are Lie groups, \cite{abraham1978foundations, arnold1989mathematical, marsden1982group,vizman2008geodesic}, as we shall confine ourselves to within the present paper. In the case that the configuration space of the system is a Lie group, say $K$, the Lagrangian dynamics on the (left trivialized) tangent bundle $TK$ is determined by the the Euler-Lagrange equations 
\begin{equation}\label{EL-on-TK}
\frac{d}{dt}\frac{\delta\mathfrak{L}}{\delta\a}   = T_e^\ast L_a \frac{\delta\mathfrak{L}}{\delta a} -\ad_{\a}^{\ast}%
\frac{\delta\mathfrak{L}}{\delta\a}
\end{equation}
generated by a Lagrangian $\mathfrak{L}:TK\to \B{R}$, $\mathfrak{L}=\mathfrak{L}(a,\a)$, where $L$ refers to the left regular representation of $K$ on itself, and $\ad^\ast$ the (left) coadjoint action of $K$ on $\G{K}^\ast$. Accordingly, the symmetry of $K$ reduces \eqref{EL-on-TK} to the Euler-Poincar\'e equations
\begin{equation} \label{ep}
\frac{d}{dt}\frac{\delta\mathfrak{L}}{\delta\a}   =-\ad_{\a}^{\ast}%
\frac{\delta\mathfrak{L}}{\delta\a}
\end{equation}
associated to the reduced Lagrangian function(al) $\mathfrak{L}:\G{K}\to \B{R}$, $\mathfrak{L} = \mathfrak{L} (\a)$, on the Lie algebra $\G{K}$ of $K$, \cite{MarsdenRatiu-book}.

One feasible strategy to study dynamics (in the Lagrangian, or Hamiltonian framework) is to realize the configuration space as a semi-direct product, \cite{n2001lagrangian,MaMiOrPeRa07,MarsRatiWein84}. This, then, allows to realize the dynamical system under investigation as a composition of two simpler subsystems, one of which effecting the other. In the presence of such a nontrivial action, of one subsystem onto the other, the joint system (under investigation) becomes more than a mere composition of its building blocks. More precisely, the (Euler-Poincar\'{e}) equations of motion of the joint system contains additional terms which do not exist in the equations of motions of the individual subsystems. As such, it is possible to capture a dynamics governed by a rather complicated set of (Euler-Poincar\'{e}) equations in terms of two much simpler dynamics governed by much simpler equations. This strategy has proved itself to be very productive. Indeed, the heavy top, ideal compressible fluids, MHD, and many other physical systems have been successfully studied within the semi-direct product framework, \cite{holm1998euler,Rati80}. 

However, although there are many examples that fit into the semi-direct product framework, it is more realistic to expect the building blocks of a given physical system to be in mutual interaction, rather than a one-way relationship. As a result, the semi-direct product strategy has been taken a step forward recently in \cite{EsenSutl17}, by considering the configuration space as a double cross product \cite{Maji90,Majid-book}. 

A group $K$ is called a ``double cross product'' of a pair of groups $(G,H)$ if $K\cong G\times H$ as sets. Then, the group structure on $K$ is determined by those in $G$ and $H$, together with their mutual actions $\rt:H\times G\to G$ of $H$ on $G$, and $\lt:H\times G \to H$ of $G$ on $H$. This group structure built on $G\times H$ is denoted by $G\bowtie H$, in order to emphasize the mutual actions. If, in particular, the (right) action of $G$ on $H$, or the (left) action of $H$ on $G$ is trivial, then $K$ becomes a semi-direct product of the pair $(G,H)$; namely $K\cong G \rtimes H$, or $K \cong G\ltimes H$ respectively. Furthermore, the Lie algebra $\G{K}$ of a double cross product group $K\cong G\bowtie H$ is a ``double cross sum'' Lie algebra, namely; $\G{K}\cong \G{g}\oplus \G{h}$, with $\G{g}$ and $\G{h}$ being the Lie algebras of $G$ and $H$ respectively. Just as in the group case, this time the Lie algebra structure on $\G{K}$ is completely determined by those in $\G{g}$ and $\G{h}$, and their mutual actions $\rt:\G{h}\ot \G{g}\to\G{g}$ and $\lt:\G{h}\ot \G{g}\to\G{h}$. Just as in the group level, the Lie algebra thus constructed on $\G{g}\oplus \G{h}$ is also denoted by $\G{g}\bowtie \G{h}$ to emphasize the mutual actions.

Just as it is in the semi-direct product theory, realizing a dynamical system as a double cross product means to consider the system under investigation as a composition of two subsystems; but this time allowing a mutual interaction between them, wherein the novelty of this approach lies. In addition to the individual dynamics of the two subsystems, this time their (nontrivial) actions on each other contributes to the dynamics of the joint system as well. 

In the level of equations of motion, in turn, the contributions of the subsystems into the total dynamics of the joint system manifest themselves in the (Euler-Poincar\'{e}) equations. As such, the realization of the configuration space of a given dynamical system as a double cross product gives a larger (compared to the semi-direct product realization) number of terms which are associated to the mutual actions of the subsystems. It is precisely this feature that gives the double cross product theory a non-trivial advantage over the semi-direct product theory. Namely; it recovers the semi-direct product theory simply by considering one of the actions of the subsystems onto the other to be trivial, and it thus qualifies to encompass the examples that fall beyond the semi-direct product theory.

More precisely; given a double cross product group $G\bowtie H$, the Euler-Lagrange equations associated to a Lagrangian $\mathfrak{L}:T(G\bowtie H) \to \B{R}$, $\mathfrak{L}=\mathfrak{L}(g,h;\xi,\eta)$ on the (left trivialized) tangent bundle $T(G\bowtie H) \cong (G\bowtie H)\ltimes (\G{g}\bowtie \G{h})$ was obtained in \cite{EsenSutl17} as
\begin{align*}
& \frac{d}{dt}\frac{\delta\mathfrak{L}}{\delta\xi}   = \underset{H\, \text{action on}\, \G{g}^\ast}{\underbrace{\Big(T_e^\ast L_g \frac{\delta\mathfrak{L}}{\delta g}\Big)\overset{\ast}{\lt} h}} + T_e^\ast \s_h \frac{\delta\mathfrak{L}}{\delta h} -\ad_{\xi}^{\ast}%
\frac{\delta\mathfrak{L}}{\delta\xi} + \underset{\G{h}\, \text{action on}\, \G{g}^\ast}{\underbrace{\frac{\delta\mathfrak{L}}{\delta \xi}\overset{\ast}{\lt} \eta}} + \G{a}^\ast_\eta\frac{\delta\mathfrak{L}}{\delta \eta} \\
& \frac{d}{dt}\frac{\delta\mathfrak{L}}{\delta\eta}   = T_e^\ast L_h \frac{\delta\mathfrak{L}}{\delta h} - \ad_{\eta}^{\ast}%
\frac{\delta\mathfrak{L}}{\delta\eta} - \underset{\G{g}\, \text{action on}\, \G{h}^\ast}{\underbrace{\xi \overset{\ast}{\rt} \frac{\delta\mathfrak{L}}{\delta\eta}}} - \G{b}^\ast_\xi\frac{\delta\mathfrak{L}}{\delta \xi}
\end{align*}
where, for any $(g,h)\in G\bowtie H$ and any $(\xi,\eta)\in \G{g}\bowtie \G{h}$, $\s_h:G\to H$ is given by $\s_h(g) := h\lt g$, $\G{a}_\eta^\ast$ is the transpose of $\G{a}_\eta:\G{g} \to \G{h}$, $\G{a}_\eta(\xi):=\eta\lt\xi$, and similarly $\G{b}_\xi^\ast$ is the transpose of $\G{b}_\xi:\G{h} \to \G{g}$ given by $\G{b}_\xi(\eta):=\eta\rt\xi$. Accordingly, the reduction with respect to the $G\bowtie H$ action yielded the Euler-Poincar\'e equations 
\begin{align}
\begin{split}\label{mEP-1}
\frac{d}{dt}\frac{\delta\mathfrak{L}}{\delta\xi}  &  =-\ad_{\xi}^{\ast}%
\frac{\delta\mathfrak{L}}{\delta\xi}+\frac{\delta\mathfrak{L}}{\delta\xi
}\overset{\ast}{\vartriangleleft}\eta+\mathfrak{a}_{\eta}^{\ast}\frac
{\delta\mathfrak{L}}{\delta\eta}\\
\frac{d}{dt}\frac{\delta\mathfrak{L}}{\delta\eta}  &  =-\ad_{\eta}^{\ast}%
\frac{\delta\mathfrak{L}}{\delta\eta}-\xi\overset{\ast}{\vartriangleright
}\frac{\delta\mathfrak{L}}{\delta\eta}-\mathfrak{b}_{\xi}^{\ast}\frac
{\delta\mathfrak{L}}{\delta\xi}
\end{split}
\end{align}
generated by the reduced Lagrangian $\mathfrak{L}:\G{g}\bowtie \G{h} \to \B{R}$, $\mathfrak{L}=\mathfrak{L}(\xi,\eta)$ on the Lie algebra $\G{g}\bowtie \G{h} $ of $G\bowtie H$.

Let us also note that the double cross product theory has already found applications in plasma theory \cite{esen2017hamiltonian}, and in thermodynamical processes \cite{esen2019,pavelka2018multiscale, vagner2019multiscale}. On the other hand, the application of the theory to the Hamiltonian and Lie-Poisson dynamics was achieved successfully in \cite{EsSu16}, whereas the discrete systems in the level of Lie groupoids was studied (from the point of view of the double cross product Lie groupoids) in \cite{esen2018matched}. Finally, even though the double cross product dynamics provides an answer to the coupling problem of Lagrangian systems, even extending the semi-direct product theory, it is far from being the complete solution; since there are Lagrangian theories other then the Euler-Poincar\'{e} dynamics. 

Now, the Lagrange equations are second order differential equations. Therefore, in order to recast a higher order differential equation in the Lagrangian framework, one needs to study the system over the higher order tangent bundles of the configuration space. In this case, the Lagrangian function depends on acceleration (and possibly the terms of higher order) in addition to the position and the velocity, \cite{de2011generalized}. The study of the higher order Lagrangian formalism, and its Hamiltonian counterpart, goes back to 1850's \cite{ostrogradsky1850memoires}. The reduction of the higher order Lagrangian systems under symmetries, on the other hand, have been investigated only recently in \cite{colombo2014higher,de1994symplectic,gay2011higher}. In particular, the higher order Euler-Poincar\'e equations associated to a (reduced) Lagrangian $\mathfrak{L} = \mathfrak{L}(\xi,\dot{\xi},\ldots, \xi^{(k-1)})$ on $T^kG/G \cong \G{g}^{\oplus\,k}$ was obtained as
\begin{equation}
\left( \frac{d}{dt}+\ad_{\xi}^{\ast}\right)  \left(\sum_{j=0}^{k-1}\,(-1)^j\,\frac{d^j}{dt^j}\frac{\d\mathfrak{L}}{\d \xi^{(j)}}\right)=0
	\label{soep>1}
\end{equation}
in \cite{gay2011higher}. In the case of $k=2$, as we shall study below, \eqref{soep>1} appears to be
	\begin{equation}
	\left( \frac{d}{dt}+\ad_{\xi}^{\ast}\right) \left( \frac{\delta \G{L}}{\delta \xi}%
	-\frac{d}{dt}\left( \frac{\delta \G{L}}{\delta\dot{\xi}}\right) \right) =0.
	\label{soep-1}
	\end{equation}
Taking a double cross product group $G\bowtie H$, we shall observe in detail how the 2nd order Euler-Poincar\'e equations \eqref{soep-1} is built on those on $G$, and on $H$, and the additional terms associated to the mutual actions of $G$ and $H$ on each other.

One other advantage of the double cross product theory is the theory's being a feasible avenue to study the reduction theory. As a result, based on the observation that
\[
TTG \cong \G{g}\bowtie T^2G
\]
we shall be able to derive the 2nd order Euler-Lagrange equations on $T^2G$ from the (1st order) Euler-Lagrange equations on $TTG$ as a result of the reduction with respect to the (left) action of $\G{g}$. The (1st order) Euler-Lagrange equations on $TTG$, in turn, may be written at once from those on the tangent bundle of a double cross product group as was calculated in \cite{EsenSutl17} simply by considering $TG \cong G\ltimes \G{g}$.

\subsection*{Organization}~

The organization of the paper is as follows. 

In the following Section \ref{Sec-mp}, we review the matched pairs of Lie groups, and matched pairs of Lie algebras. More precisely; Subsection \ref{Sec-mplg} contains the fundamentals of matched pairs of Lie groups, while Subsection \ref{subsect-matched-Lie-alg} is about the matched pairs of Lie algebras.

The main objects of study of the present article is introduced and studied in Section \ref{Sec-mpitg}. In Subsection \ref{subsect-tangent-group} we present the (1st order) tangent group $TG$, and its double cross product decomposition. In Subsection \ref{sec-sotg}, on the other hand, we recall the 2nd order tangent group $T^2G$, its realization as a 2-cocycle extension, and its double cross product decomposition. Finally, in Subsection \ref{itg}, we present the iterated tangent group $TTG$, and more importantly, its double cross product decomposition into $T^2G$ and $\G{g}$. 

Section \ref{mpcont}, is the section where the main results of the paper lie. Subsection \ref{fold-sec} contains a brief review of the 1st order (matched) Lagrangian dynamics. Then, in Subsection \ref{sold-sec}, there comes the 2nd order Euler-Lagrange equations on the 2nd order tangent group, as a result of the reduction of the 1st order Lagrangian dynamics on the iterated tangent bundle. Finally, we present the 2nd order Euler-Lagrange equations on the 2nd order tangent group of a double cross product group in Subsection \ref{msoELe}.

The illustrations, of the calculations done in Subsection \ref{msoELe}, comes in Section \ref{sect-illustrations}. In Subsection \ref{subsect-2-splines}, we focus on Riemannian 2-splines, while in Subsection \ref{subsect-2nd-order-3D} we discuss $3D$-systems.

\subsection*{Notations and Conventions}~

Throughout the text $G$ and $H$ will stand for Lie groups, with Lie algebras $\G{g}$ and $\G{h}$ respectively. The (linear) duals of $\G{g}$ and $\G{h}$, on the other hand, will be denoted by $\G{g}^\ast$ and $\G{h}^\ast$. We shall make use of the notation
\begin{equation}\label{G}
\begin{split}
g,\tilde{g},\tilde{\tilde{g}} \in G,\qquad \xi,\tilde{\xi},\tilde{\tilde{\xi}} \in \mathfrak{g},\qquad \mu,\tilde{\mu}
,\tilde{\tilde{\mu}} \in\mathfrak{g}^{\ast} 
\\
h,\tilde{h},\tilde{\tilde{h}} \in H,\qquad \eta,\tilde{\eta},\tilde{\tilde{\eta}} \in \mathfrak{g},\qquad \nu,\tilde{\nu}
,\tilde{\tilde{\nu}} \in\mathfrak{h}^{\ast}
\end{split}  
\end{equation}
for the generic elements. As for the representations; 
\[
L:G\times G \to G, \qquad (g,\tilde{g})\mapsto L_g\,\tilde{g}:=g\tilde{g}
\] 
will stand for the left regular representation of $G$ on itself, while the right regular representation will be denoted by 
\[
R:G\times G \to G, \qquad (g,\tilde{g})\mapsto L_{\tilde{g}}\,g:=g\tilde{g}.
\]
Similarly, 
\[
\Ad:G\times G \to G, \qquad (g,\tilde{g})\mapsto \Ad_g\,\tilde{g}:=g\tilde{g}g^{-1}
\]
will be the (left) adjoint representation of $G$ on itself. The same symbol will also be used to denote the (left) adjoint representation of $G$ on its Lie algebra $\G{g}$, namely 
\begin{equation}\label{adj-G-g}
\Ad:G\times \G{g} \to \G{g}, \qquad (g,\xi)\mapsto \Ad_g\,\xi,
\end{equation}
which in turn, induces the (left) coadjoint representation 
\[
\Ad^\ast:G\times \G{g}^\ast \to \G{g}^\ast 
\]
of $G$ on $\G{g}^\ast$ via
\begin{equation}\label{dist*}
\left\langle \Ad_{g^{-1}}^{\ast}\mu,\xi\right\rangle =\left\langle
\mu,\Ad_{g}\xi\right\rangle.   
\end{equation}

The infinitesimal counterpart of \eqref{adj-G-g} is the (left) adjoint representation 
\[
\ad:\G{g}\ot\G{g}\to \G{g}, \qquad \xi\ot \tilde{\xi}\mapsto \ad_\xi\,\tilde{\xi}:=[\xi, \tilde{\xi}]
\]
of $\G{g}$ on itself. Accordingly, we shall denote the (left) coadjoint representation of $\G{g}$ on $\G{g}^\ast$ as
\[
\ad^\ast:\G{g}\ot\G{g}^\ast\to \G{g}^\ast, \qquad \xi\ot \mu\mapsto \ad^\ast_\xi\,\mu,
\]
which is determined by
\begin{equation}\label{coadj-g-g*}
\left\langle \ad^\ast_\xi\mu,\tilde{\xi}\right\rangle =-\left\langle
\mu,\ad_{\xi}\tilde{\xi}\right\rangle.   
\end{equation}

\section{Matched pairs of Lie Groups, and matched pairs of Lie Algebras} \label{Sec-mp}

In this section we recall briefly the very basics on the double cross product Lie groups, and their Lie algebras. More precisely; in Subsection \ref{Sec-mplg} we shall recall the notion of a ``matched pair'' of Lie groups, and thus a double cross product Lie group built on a matched pair of Lie group. In Subsection \ref{subsect-matched-Lie-alg}, then, we shall turn towards the infinitesimal counterpart of this theory, namely the double cross sum Lie algebras built on the matched pairs of Lie algebras. For further details the reader may consult to the incomplete list \cite{LuWe90,Maji90-II,Maji90,Majid-book,Ta81,Zhan10} of references. 

\subsection{Matched pairs of Lie groups} \label{Sec-mplg}~

Let $(G,H)$ be a pair of Lie groups equipped with the mutual actions (left $H$-action on $G$, and right $G$-action on $H$)
\begin{equation}\label{rho}
\rt :H\times G\rightarrow G,\quad \left( h,g\right) \mapsto
h\rt g, \qquad \lt:H\times G\rightarrow H,\quad \left(
h,g\right) \mapsto h\lt g
\end{equation}
which are subject to 
\begin{align*}
& h\rt (g\tilde{g}) = (h\rt g)((h\lt g)\rt\tilde{g}),\\ 
& (h\tilde{h}) \lt g = (h\lt(\tilde{h}\rt g))(\tilde{h}\lt g).
\end{align*}
Then, the pair $(G,H)$ is called a ``matched pair'' of Lie groups. In this case, there is a group structure on the cartesian product $G\times H$ determined by
\begin{equation} \label{Gr-m}
\left( g ,h \right) \left(
\tilde{g},\tilde{h}\right) =\Big( g \left( h\rt
\tilde{g}\right) ,\left( h \lt \tilde{g} \right) \tilde{h}\Big),
\end{equation}
as was observed in \cite[Prop. 6.2.15]{Majid-book}. The group $G\times H$ given by \eqref{Gr-m} is called the ``double cross product'' of $G$ and $H$, and is denoted by $G\bowtie H$ to emphasize the mutual actions.

Conversely, as was also observed in \cite[Prop. 6.2.15]{Majid-book}, given a group $K$ with two subgroups
\[
G\hookrightarrow K \hookleftarrow H
\] 
so that the multiplication on $K$ induces an isomorphism $K\cong G\times H$ as (topological) sets,
\[
G\times H \to K, \qquad (g,h)\mapsto gh,
\]
the pair $(G,H)$ becomes a matched pair of Lie groups, and $K\cong G\bowtie H$ as groups. In this case, the mutual actions of $G$ and $H$ on each other are obtained from
\begin{equation}\label{mutual-actions}
hg = (h\rt g)(h\lt g).
\end{equation}

The double cross product construction generalizes to that of the semi-direct product construction by setting one of the mutual actions to be trivial. More precisely, given a matched pair of groups $(G,H)$, with trivial (left) $H$-action, the double cross product group $G\bowtie H$ is nothing but the semi-direct product group $G\ltimes H$. Similarly, in case of the trivial (right) $G$-action $G\bowtie H$ reduces to the the semi-direct product $G\rtimes H$.

Let us finally remark that although we follow the notation and the terminology of \cite{Maji90,Majid-book}, the double cross product construction goes back to \cite{Sz50,Sz58,szep1962sulle,Za42}; after which it was named as the Zappa-Sz\'ep product in \cite{Br05}. The very same construction appeared also in \cite{KoMa88} under the name of the twilled extension, and in \cite{LuWe90} as the double Lie group.

\subsection{Matched pairs of Lie algebras}\label{subsect-matched-Lie-alg}~

A pair of Lie algebras $(\G{g},\G{h})$ is called a ``matched pair of Lie algebras'' if their mutual actions 
\begin{equation} \label{Lieact}
\rt:\G{h}\ot\G{g}\rightarrow \G{g},\quad \eta\ot \xi \mapsto \eta\rt \xi, \qquad \lt:\G{h}\ot\G{g}\rightarrow \G{h}, \quad \eta\ot \xi \mapsto \eta\lt \xi,
\end{equation}
on each other satisfy
\begin{align*}
& \eta \rt [\xi,\tilde{\xi}] = 
[\eta \rt\xi,\tilde{\xi}]+[\xi,\eta \rt \tilde{\xi}] + (\eta
\lt \xi)\rt \tilde{\xi} - (\eta \lt \tilde{\xi})\rt \xi, \\
& [\eta,\tilde{\eta}]\lt\xi = [\eta,\tilde{\eta}\lt\xi ] + 
[\eta\lt\xi ,\tilde{\eta}] +\eta\lt (\tilde{\eta}\rt \xi) - 
\tilde{\eta}\lt (\eta\rt \xi ).
\end{align*}
It follows from \cite[Prop. 8.3.2]{Majid-book} that given a matched pair of Lie algebras, there is a Lie algebra structure on $\G{g}\oplus \G{h}$ given by
\begin{equation}\label{mpla}
[ (\xi,\eta),\,(\tilde{\xi},\tilde{\eta})] =\left( [\xi,\tilde{\xi}]+\eta\rt \tilde{\xi}-\tilde{\eta}\rt \xi
,\,[\eta,\tilde{\eta}]+\eta\lt \tilde{\xi}-\tilde{\eta}\lt \xi\right).  
\end{equation}
The Lie algebra $\G{g}\oplus \G{h}$ given by \eqref{mpla} is called a ``double cross sum'' Lie algebra, and  it is denoted by $\G{g}\bowtie \G{h}$. Conversely, if $\G{K}$ is a Lie algebra, with Lie subalgebras
\[
\G{g}  \hookrightarrow\G{K} \hookleftarrow \G{h}
\]
so that the summation on $\G{K}$ induces an isomorphism $\G{K} \cong \G{g}\bowtie \G{h}$ as vector spaces,
\[
\G{g}\bowtie \G{h} \to \G{K}, \qquad (\xi,\eta)\mapsto \xi + \eta, 
\]
then the pair $(\G{g},\G{h})$ is a matched pair, and $\G{K} \cong \G{g}\bowtie \G{h}$ as Lie algebras.

Let us note also that if $(G,H)$ be a matched pair of Lie groups, then their Lie algebras $(\G{g},\G{h})$ make a matched pair of Lie algebras, and the Lie algebra of the double cross product Lie group $G\bowtie H$ is the double cross sum Lie algebra $\G{g}\bowtie \G{h}$.

Just as it is for the matched pairs of Lie groups, a double cross sum Lie algebra reduces to a semi-direct sum Lie algebra if one of the actions of the components is assumed to be trivial. More precisely, if the left $\G{h}$-action on $\G{g}$ is trivial, then $\G{g}\bowtie \G{h}$ reduces to the semi-direct sum $\G{g}\ltimes \G{h}$. Similarly, in case the right $\G{g}$- action on $\G{h}$ is trivial, then $\G{g}\bowtie \G{h}$ becomes the semi-direct sum Lie algebra $\G{g}\rtimes \G{h}$.

\section{The 2nd order and the iterated tangent groups} \label{Sec-mpitg}

The present section is about the main objects of study of this note. In Subsection \ref{subsect-tangent-group} we begin with a quick review of the (1st order) tangent group $TG$, and its double cross product decomposition. Then in Subsection \ref{sec-sotg} below we record the group structure of the 2nd order tangent group $T^2G$, its realization as a 2-cocycle extension, and its double cross product decomposition. We conclude the section with the iterated tangent group $TTG$ in Subsection \ref{itg}, where we observe its double cross product decomposition into $T^2G$ and $\G{g}$, as well as two other presentations.

\subsection{The (1st order) tangent group}\label{subsect-tangent-group}~

Let us take a quick tour around the group structure on the tangent bundle $TG$. To this end, we recall the (left) trivialization 
\begin{equation}\label{trTG}
tr:TG\rightarrow G\ltimes\G{g}_1, \qquad V_{g}\mapsto (g,T_{g}L_{g^{-1}}V_{g})=(g,\xi),   
\end{equation}
via which the tangent bundle $TG$ may be endowed with the semi-direct product group structure on  $G\ltimes\G{g}_1$, which in turn is given explicitly by
\begin{equation}\label{tgtri}
\left( g,\xi^{(1)}\right) \left( \tilde{g},\tilde{\xi}^{(1)} \right) =\left(
g\tilde{g},\tilde{\xi}^{(1)}+\Ad_{\tilde{g}^{-1}}\xi^{(1)}\right),
\end{equation}
for any $\xi^{(1)}, \tilde{\xi}^{(1)} \in \G{g}_1 = \G{g}$. Accordingly, the Lie algebra of $TG\cong G\ltimes \G{g}$ is the semi-direct sum Lie algebra $\G{g}_2\ltimes \G{g}_3:=\G{g}\ltimes \G{g}$ whose Lie bracket being
\begin{equation}
[(\xi^{(2)},\xi^{(3)}),(\tilde{\xi}^{(2)},\tilde{\xi}^{(3)})] = ([\xi^{(2)},\tilde{\xi}^{(2)}], \ad_{\xi^{(3)}}\tilde{\xi}^{(2)}-\ad_{\tilde\xi^{(3)}} {\xi}^{(2)})
\end{equation}
for any $\xi^{(2)}, \tilde{\xi}^{(2)} \in \G{g}_2 = \G{g}$, and any $\xi^{(3)}, \tilde{\xi}^{(3)} \in \G{g}_3 = \G{g}$. We note that the indices the Lie algebra $\G{g}$ serve to distinguish the copies of it, since we shall need it in the sequel.

Leaving the further details on the (1st order) tangent group to \cite{esen2014tulczyjew, hindeleh2006tangent, KolaMichSlov-book, marsden1991symplectic, michor2008topics, Rati80}, we now recall the double cross product decomposition of it from \cite[Prop. 2.3]{EsenSutl17}. 

Accordingly, given a matched pair of Lie groups $(G,H)$, their tangent groups $(TG,TH)$ is also a matched pair of Lie groups. Furthermore, 
\begin{equation}\label{TGH-TG-TH}
T(G\bowtie H) \cong TG\bowtie TH
\end{equation}
as Lie groups. In this case, the mutual actions of the tangent groups may be computed to be
\begin{align}
& (h,\eta^{(1)})\rt (g,\xi^{(1)}) = \Big(h\rt g, (h\lt g)\rt \xi^{(1)} + (h\lt g)\rt T_gL_{g^{-1}}(\eta^{(1)} \rt g)\Big), \label{TH-on-TG}\\
& (h,\eta^{(1)})\lt (g,\xi^{(1)}) = \Big(h\lt g, \eta^{(1)}\lt g + T_{(h\lt g)}L_{(h\lt g)^{-1}} \Big((h\lt g)\lt \big(\xi^{(1)} + T_gL_{g^{-1}}(\eta^{(1)} \rt g)\big)\Big)\Big),\label{TG-on-TH}
\end{align}
and the group structure on \eqref{TGH-TG-TH} is thus given by
\[
(g,h ,\xi ,\eta )(\tilde{g},\tilde{h},\tilde{\xi},\tilde{\eta}) =(\tilde{\tilde{g}},\tilde{\tilde{h}},\tilde{\tilde{\xi}},\tilde{\tilde{\eta}}),
\]
where
\begin{align*}
& \tilde{\tilde{g}}=g (h\rt \tilde{g}), \qquad \tilde{\tilde{h}}=(h \lt \tilde{g})\tilde{h}, \qquad 
\tilde{\tilde{\xi}}=\tilde{\xi}+\tilde{h}^{-1} \rt \left(\Ad_{\tilde{g}^{-1}}\xi +T_{\tilde{g}}L_{\tilde{g}^{-1}}(\eta  \rt \tilde{g})\right), \\
&\tilde{\tilde{\eta}}=\tilde{\eta}+T_{\tilde{h}^{-1}}R_{\tilde{h}}\left( \tilde{h}^{-1} \lt (\Ad_{\tilde{g}^{-1}}\xi +T_{\tilde{g}}L_{\tilde{g}^{-1}}(\eta  \rt \tilde{g}) ) \right)+\Ad_{\tilde{h}^{-1}}(\eta  \lt \tilde{g}).
\end{align*}

\subsection{The 2nd order tangent group}\label{sec-sotg}~

Along the lines of \cite{gay2012invariant,gay2011higher}, the $k$th order tangent bundle $\tau^k:T^kQ\to Q$ of a manifold $Q$ is defined to be the equivalence classes of twice differentiable curves in $Q$ with respect to the equivalence relation that relates two curves with identical derivatives at 0 up to $k$th order. More precisely, two curves $p(t),q(t)\in Q$ are set to be equivalent if 
\[
p(0) = q(0), \qquad p'(0) = q'(0), \qquad \ldots \qquad p^{(k)}(0) = q^{(k)}(0).
\]
Accordingly, $[p]$ being the equivalence class of the curve $p(t)\in Q$, the bundle projection is defined to be
\[
\tau^k:T^kQ\to Q, \qquad \tau^k([p]) := p(0).
\]
Refering the reader to \cite{abrunheiro2011cubic, colombo2011geometry, colombo2013optimal, colombo2014unified} for higher order tangent bundles, we now proceed onto the case $Q = G$, a Lie group. In this case, the $k$th order tangent group $T^kG$ has the structure of a Lie group, see for instance \cite{gay2012invariant,KolaMichSlov-book,Vizm13}, via
\[
[p][q] := [pq].	
\] 
Given a curve $p(t)\in G$, now, $\d^\ell p := p^{-1}p'$ being the ``left logarithmic derivative'', the left trivialization (compare with the right-handed version in\cite{Vizm13})
\[
tr^2:T^2G\to G\times \G{g}\times \G{g}, \qquad  [p] \mapsto \Big(p(0), (\d^\ell p)(0), (\d^\ell p)'(0)\Big)
\]
allows to endow $T^2G$ with the group structure given by
\begin{equation} \label{GrT2G}
(g,\xi,\dot{\xi}) (\tilde{g},\tilde{\xi},\dot{\tilde{\xi}}) = (g\tilde{g}, \tilde{\xi} +\Ad_{\tilde{g}^{-1}}\xi,\dot{\tilde{\xi}} +\Ad_{\tilde{g}^{-1}}\dot{\xi} -\ad_{\tilde{\xi}}\Ad_{\tilde{g}^{-1}}\xi).
\end{equation}
Accordingly, the unit element is $(e,0,0) \in G \times \G{g} \times  \G{g}$, and the inversion is given by
\[
(g,\xi,\dot{\xi})^{-1} = (g^{-1}, -\Ad_g\xi, -\Ad_g\dot{\xi}).
\]
Once again, in order to distinguish the copies of the Lie algebra $\G{g}$, we shall make use of the notation $G \times \G{g} \times  \dot{\G{g}}$.

The next proposition sheds further light on the (group) structure of $T^2G$, see also \cite{Vizm13}.

\begin{proposition}
The 2nd order tangent group $T^2G$ is a 2-cocycle extension of $TG$, by $\dot{\G{g}}$. In short, 
\begin{equation}\label{T2G-cocycle-ext}
T^2G \cong TG \ltimes_\vp \dot{\G{g}}
\end{equation}
\end{proposition}

\begin{proof}
It follows at once that 
\[ \lt:\G{\dot{g}}\times TG \to \G{\dot{g}}, \qquad (\dot{\xi},(g,\xi))\mapsto 
\dot{\xi} \lt (g,\xi) := \Ad_{g^{-1}}\dot{\xi}
\]
determines a right action of $TG$ on $\dot{\G{g}}$, and that, in view of the identification $T^2G\cong G\times \G{g}\times \dot{\G{g}}$ above (via the left trivialization), the group structure may be expressed as
\[
(g,\xi,\dot{\xi}) (\tilde{g},\tilde{\xi},\dot{\tilde{\xi}}) = \Big((g,\xi)(\tilde{g},\tilde{\xi}); \dot{\tilde{\xi}} + \dot{\xi}\lt (\tilde{g},\tilde{\xi}) + \vp((g,\xi),(\tilde{g},\tilde{\xi}))\Big),
\]
where
\begin{equation}\label{vp-map}
\vp: TG \times TG \to \dot{\G{g}}, \qquad \vp((g,\xi),(\tilde{g},\tilde{\xi})) := -\ad_{\tilde{\xi}}\Ad_{\tilde{g}^{-1}}\xi.
\end{equation}
On the other hand, it takes a routine calculation to observe that \eqref{vp-map} satisfies the (cocycle) condition
\begin{align*}
& d\vp((g,\xi), (\tilde{g},\tilde{\xi}),(\tilde{\tilde{g}},\tilde{\tilde{\xi}})) := \\
&\vp((\tilde{g},\tilde{\xi}),(\tilde{\tilde{g}},\tilde{\tilde{\xi}})) - \vp((g,\xi)(\tilde{g},\tilde{\xi}),(\tilde{\tilde{g}},\tilde{\tilde{\xi}})) + \vp((g,\xi),(\tilde{g},\tilde{\xi})(\tilde{\tilde{g}},\tilde{\tilde{\xi}})) - \vp((g,\xi),(\tilde{g},\tilde{\xi}))\lt (\tilde{\tilde{g}},\tilde{\tilde{\xi}}) = 0,
\end{align*}
for any $(g,\xi), (\tilde{g},\tilde{\xi}),(\tilde{\tilde{g}},\tilde{\tilde{\xi}}) \in TG$. That is, \eqref{vp-map} determines a 2-cocycle in the group cohomology of $TG$, with coefficients in $\dot{\G{g}}$; or in short, $\vp \in H^2(TG,\dot{\G{g}})$. The claim thus follows.
\end{proof}

We shall conclude the present subsection with the double cross product decomposition of the 2nd order tangent group. 

\begin{proposition} \label{prop-msotg}
If $(G,H)$ is a matched pair of Lie groups, then so is $(T^2G, T^2H)$. Moreover, 
\begin{equation}\label{T2GH-matched}
T^2(G\bowtie H) \cong T^2G\bowtie T^2H
\end{equation}
as Lie groups.
\end{proposition}

\begin{proof} 	
In view of \cite[Prop. 6.2.15]{Majid-book}, we begin with the inclusions 
\begin{align} \label{inc-T2G}
\begin{split}
& T^2G \hookrightarrow T^2(G\bowtie H), \qquad (g,\xi,\dot{\xi}) \mapsto \Big((g,e), (\xi,0), (\dot{\xi},0)\Big), \\
& T^2H \hookrightarrow T^2(G\bowtie H), \qquad (h,\eta,\dot{\eta}) \mapsto \Big((e,h), (0,\eta), (0,\dot{\eta})\Big).
\end{split}
\end{align}
Next, recording the group multiplication
\begin{align*}
&(g,h,\xi,\eta,\dot{\xi},\dot{\eta})(\tilde{g},\tilde{h},\tilde{\xi},\tilde{\eta},\dot{\tilde{\xi}},\dot{\tilde{\eta}}) = \\
& \Big(g(h \triangleright \tilde{g}),(h \triangleright \tilde{g})\tilde{h}, \tilde{\xi}+\tilde{h}^{-1}\triangleright \tau,\tilde{\eta}+T_{\tilde{h}^{-1}}R_{\tilde{h}}(\tilde{h}^{-1} \triangleleft \tau)+\Ad_{\tilde{h}^{-1}}(\eta \triangleleft \tilde{g}),\\
& \hspace{6cm} \dot{\tilde{\xi}}+ \tilde{h}^{-1}\triangleright \dot{\tau}, 
\dot{\tilde{\eta}}+T_{\tilde{h}^{-1}}R_{\tilde{h}}( \tilde{h}^{-1} \triangleleft \dot{\tau})+\Ad_{\tilde{h}^{-1}}(\dot{\eta} \triangleleft \tilde{g})\Big), 
\end{align*}
where 
\begin{align*}
\tau=\Ad_{\tilde{g}^{-1}}\xi+T_{\tilde{g}}L_{\tilde{g}^{-1}}(\eta \triangleright \tilde{g}), \qquad 
\dot{\tau}=\Ad_{\tilde{g}^{-1}}\dot{\xi}+T_{\tilde{g}}L_{\tilde{g}^{-1}}(\dot{\eta} \triangleright \tilde{g}),
\end{align*}
we proceed to the multiplication map 
	\begin{align*}
	& T^2G\times T^2H \to T^2(G\bowtie H), \\
	& \Big(\big(g,\xi,\dot{\xi}\big),\big(h,\eta,\dot{\eta}\big)\Big) \mapsto \Big((g,e), (\xi,0), (\dot{\xi},0)\Big)\Big((e,h), (0,\eta), (0,\dot{\eta})\Big) =  \\
	&\Big((g,h), \Ad_{(e,h)^{-1}}(\xi,0) + (0,\eta), \vp\big(((g,e), (\xi,0)),((e,h), (0,\eta))\big) + \Ad_{(e,h)^{-1}}(\dot{\xi},0) + (0,\dot{\eta})\Big) \\
	& \qquad  =
	\Big((g,h), \Ad_{(e,h)^{-1}}(\xi,0) + (0,\eta), -\ad_{(0,\eta)}\Ad_{(e,h)^{-1}}(\xi,0) + \Ad_{(e,h)^{-1}}(\dot{\xi},0) + (0,\dot{\eta})\Big) \\
	& \qquad  =
	\Big((g,h), \Ad_{(e,h)^{-1}}(\xi,0) + (0,\eta), \big[\Ad_{(e,h)^{-1}}(\xi,0), (0,\eta)\big] + \Ad_{(e,h)^{-1}}(\dot{\xi},0) + (0,\dot{\eta})\Big) \\
	& \qquad  = \Big((g,h),(h^{-1}\rt \xi, T_{h^{-1}}R_h(h^{-1}\lt \xi)+\eta), \\
	&\hspace{2cm}\big[(h^{-1}\rt \xi, T_{h^{-1}}R_h(h^{-1}\lt \xi), (0,\eta)\big] + (h^{-1}\rt \dot{\xi}, T_{h^{-1}}R_h(h^{-1}\lt \dot{\xi}) + \dot{\eta})\Big) \\
	& \qquad  = \Big((g,h),(h^{-1}\rt \xi, T_{h^{-1}}R_h(h^{-1}\lt \xi)+\eta), \\
	&\hspace{2cm}\big( -\eta \rt (h^{-1}\rt \xi), [T_{h^{-1}}R_h(h^{-1}\lt \xi),\eta] - \eta\lt (h^{-1}\rt \xi)\big) + (h^{-1}\rt \dot{\xi}, T_{h^{-1}}R_h(h^{-1}\lt \dot{\xi}) + \dot{\eta})\Big) \\
\end{align*}
\begin{align*}
	& \qquad  = \Big((g,h),(h^{-1}\rt \xi, T_{h^{-1}}r_h(h^{-1}\lt \xi)+\eta), \\
	&\hspace{2cm}\big(h^{-1}\rt \dot{\xi} -\eta \rt (h^{-1}\rt \xi), [T_{h^{-1}}R_h(h^{-1}\lt \xi),\eta] - \eta\lt (h^{-1}\rt \xi) + T_{h^{-1}}R_h(h^{-1}\lt \dot{\xi}) + \dot{\eta}\big) \Big).
	\end{align*}
We now conclude with the the observation that the multiplication map above is invertible, by calculating its inverse explicitly. Setting
	\begin{align*}
	& T^2(G\bowtie H) \to T^2G \times T^2H, \\
&	\Big((g,h),(\xi,\eta);(\dot{\xi},\dot{\eta})\Big) \mapsto \Big((g,e),(A_1,0),(A_0,0)\Big)\Big((e,h),(0,B_1),(0,B_0)\Big),
	\end{align*}
it follows from
	\begin{align*}
	\Big((g,h),(\xi,\eta)\Big) &= \Big((g,e),(A_1,0)\Big)\Big((e,h),(0,B_1)\Big) = \Big((g,h),\Ad_{(e,h)^{-1}}(A_1,0) + (0,B_1)\Big) \\
	& = \Big((g,h),(h^{-1}\rt A_1,T_{h^{-1}}R_h(h^{-1}\lt A_1) + B_1) \Big),
	\end{align*}
that 
	\begin{equation} \label{A_B_1}
	A_1 = h\rt \xi, \qquad B_1 = \eta - T_{h^{-1}}R_h(h^{-1}\lt (h\rt \xi)) = \eta + T_hL_{h^{-1}}(h\lt \xi).
	\end{equation}
	In order to determine $A_0$ and $B_0$ we compute
\begin{align*}
	(\dot{\xi},\dot{\eta} ) & = \vp\big(((g,e),(A_1,0)), \, ((e,h),(0,B_1))\big) + \Ad_{(e,h)^{-1}}(A_0,0) + (0,B_0) \\
	& = \big[\Ad_{(e,h)^{-1}}(A_1,0), (0,B_1)\big]  + \Ad_{(e,h)^{-1}}(A_0,0) + (0,B_0) \\
& =\big[(\xi,\eta) - (0,B_1), (0,B_1)\big]  + (h^{-1}\rt A_0, T_{h^{-1}}R_h(h^{-1}\lt A_0) + B_0)  \\
		& =\big[(\xi,\eta) , (0,B_1)\big]  + (h^{-1}\rt A_0, T_{h^{-1}}R_h(h^{-1}\lt A_0) + B_0)  \\
	& = (-B_1 \rt \xi,[\eta,B_1] - B_1\lt \xi ) +  (h^{-1}\rt A_0, T_{h^{-1}}R_h(h^{-1}\lt A_0) + B_0 ) \\
	& = (h^{-1}\rt A_0-B_1 \rt \xi,[\eta,B_1] - B_1\lt \xi + T_{h^{-1}}R_h(h^{-1}\lt A_0) + B_0 ). 
	\end{align*}
Accordingly,
	\begin{equation} \label{A_0B_0}
	A_0 = h\rt (B_1 \rt \xi) + h\rt \dot{\xi}, \qquad B_0 = \dot{\eta} + B_1\lt \xi - T_{h^{-1}}R_h(h^{-1}\lt A_0) + [B_1,\eta].
	\end{equation}
Finally we substitute $A_1$ and $B_1$ of \eqref{A_B_1} into \eqref{A_0B_0} to obtain 
	\begin{equation}
	\begin{split}
	A_0 &= h\rt (\eta \rt \xi) + h\rt \big(T_hL_{h^{-1}}(h\lt \xi) \rt \xi\big) + h\rt \dot{\xi},
	\\
	B_0 &= \dot{\eta} + \eta\lt \xi + T_hL_{h^{-1}}(h\lt \xi)\lt \xi + T_hL_{h^{-1}}(h\lt (\eta\rt \xi))  \\ & \qquad +
	T_hL_{h^{-1}}(h\lt (T_hL_{h^{-1}}(h\lt \xi)\rt \xi)) + T_hL_{h^{-1}}(h\lt \dot{\xi}) + [T_hL_{h^{-1}}(h\lt \xi),\eta].
	\end{split}
	\end{equation}
\end{proof}

Now the same line of thought as \eqref{mutual-actions}, namely 
\begin{equation}
\Big((e,h),(0,\eta),(0,\dot{\eta})\Big)\Big((g,e),(\xi,0),(\dot{\xi},0)\Big) = \bigg\{\Big(h,\eta,\dot{\eta}\Big)\rt\Big(g,\xi,\dot{\xi}\Big)\bigg\}\bigg\{\Big(h,\eta,\dot{\eta}\Big)\lt\Big(g,\xi,\dot{\xi}\Big)\bigg\},
\end{equation}
yields the mutual actions of $T^2G$ and $T^2H$ on each other. More explicitly, the left action of $T^2H$ on $T^2G$ is given by 
\begin{align}
\begin{split}
& \Big(h,\eta,\dot{\eta}\Big)\rt\Big(g,\xi,\dot{\xi}\Big) =  \bigg(h\rt g, \, (h\lt g)\rt \big(\xi + T_gL_{g^{-1}}(\eta \rt g)\big), \\ 
& \hspace{2cm} (h\lt g) \rt \Big((\eta\lt g) \rt \big(\xi + T_gL_{g^{-1}}(\eta \rt g)\big)\Big)  \\
& \hspace{2cm} +(h\lt g) \rt \bigg(\Big\{T_{h\lt g}L_{(h\lt g)^{-1}}\Big((h\lt g) \lt \big(\xi + T_gL_{g^{-1}}(\eta\rt g)\big) \Big)\Big\} \rt \big(\xi + T_gL_{g^{-1}}(\eta \rt g)\big)\bigg) \\ 
 & \hspace{2cm} +(h\lt g) \rt \big(\dot{\xi} + T_gL_{g^{-1}}(\dot{\eta} \rt g) -\ad_\xi\big(T_gL_{g^{-1}}(\eta \rt g)\big) + (\eta\lt g)\rt \xi  \big)\bigg),
\end{split}
\end{align}
whereas the right action of $T^2G$ on $T^2H$ is
\begin{align}
\begin{split}
& \Big(h,\eta;\dot{\eta}\Big)\lt\Big(g,\xi,\dot{\xi}\Big) = \bigg( h\lt g,\, \eta \lt g + T_{h\lt g}L_{(h\lt g)^{-1}}\Big((h\lt g) \lt \big(\xi + T_gL_{g^{-1}}(\eta\rt g)\big)\Big),\\ 
& \hspace{2cm} \dot{\eta} \lt g + (\eta\lt g)\lt \xi + (\eta \lt g)\lt \big(\xi + T_gL_{g^{-1}}(\eta \rt g)\big)  \\
& \hspace{2cm} + T_{h\lt g}L_{(h\lt g)^{-1}}\Big((h\lt g) \lt \big(\xi + T_gL_{g^{-1}}(\eta \rt g)\big)\Big) \lt \big(\xi + T_gL_{g^{-1}}(\eta \rt g)\big) \\ 
& \hspace{2cm} + T_{h\lt g}L_{(h\lt g)^{-1}}\Big((h\lt g) \lt \Big\{(\eta \lt g)\rt\big(\xi + T_gL_{g^{-1}}(\eta \rt g)\big)\Big\}\Big) \\ 
&  \hspace{2cm}  + T_{h\lt g}L_{(h\lt g)^{-1}}\Big( (h\lt g) \lt \Big\{T_{h\lt g}L_{(h\lt g)^{-1}}\Big((h\lt g) \lt \big(\xi + T_gL_{g^{-1}}(\eta\rt g)\big)\Big)\rt \big(\xi + T_gL_{g^{-1}}(\eta\rt g)\big) \Big\}\Big) \\ 
& \hspace{2cm}  + T_{h\lt g}L_{(h\lt g)^{-1}}\Big( (h\lt g) \lt \Big\{\dot{\xi} + T_gL_{g^{-1}}(\dot{\eta}\rt g) -\ad_\xi\big(T_gL_{g^{-1}}(\eta\rt g)\big) + (\eta\lt g)\rt\xi \Big\}\Big)\bigg) \\
& \hspace{2cm}  + [T_{h\lt g}L_{(h\lt g)^{-1}}\Big((h\lt g) \lt \big(\xi + T_gL_{g^{-1}}(\eta\rt g)\big)\Big), \, \eta\lt g].
\end{split}
\end{align}

\subsection{The iterated tangent group} \label{itg} ~

We shall now turn out attention to the (twice) iterated tangent group
\begin{equation}\label{TTG-gr-structure}
TTG\cong T(TG) \cong T(G\ltimes\G{g}_1) \cong (G\ltimes\G{g}_1)\ltimes (\G{g}_2\ltimes\G{g}_3).
\end{equation}
Iterative application of the (left) trivialization  
\begin{align}\label{trTTG}
\begin{split}
& tr: TTG \to ( G\ltimes\G{g}_1)
\ltimes (\G{g}_2\ltimes\G{g}_3), \\  
& \left( V_{g},V_{\xi^{(1)}}\right) \mapsto (g,\xi^{(1)},\xi^{(2)},\xi^{(3)}):=\left(
g,\xi^{(1)},TL_{g^{-1}}V_{g},V_{\xi^{(1)}}-\left[ \xi^{(1)},TL_{g^{-1}}V_{g}\right]
\right),
\end{split}
\end{align}
yields the multiplication
\begin{equation}\label{GrTTG}
\begin{split}
&\left( g,\xi^{(1)},\xi^{(2)},\xi^{(3)}\right) \left(
\widetilde{g},\tilde{\xi}^{(1)},\tilde{\xi}^{(2)},\tilde{\xi}^{(3)}
\right) =\\ 
& \left( g\widetilde{g},\tilde{\xi}^{(1)}+\Ad_{\widetilde{g}^{-1}}\xi^{(1)},\tilde{\xi}^{(2)}+\Ad_{\widetilde{g}^{-1}}\xi^{(2)},\tilde{\xi}^{(3)}+\Ad_{\widetilde{g}^{-1}}\xi^{(3)}+[\Ad_{\widetilde{g}^{-1}}\xi^{(2)},\tilde{\xi}^{(1)}]\right),
\end{split}
\end{equation}
see for instance \cite{bertram2008differential,Vizm13}, whose identity element being $(e,0,0,0) \in TTG$. Once again, the indices of the Lie algebra $\G{g}$ serve to distinguish the identical copies.

\begin{remark}
It follows at once from \cite[Prop. 2.3]{EsenSutl17} that if $(G,H)$ is a matched pair of Lie groups, then so is $(TTG,TTH)$, since $(TG,TH)$ is. Accordingly, 
\[
TT(G\bowtie H) \cong TTG \bowtie TTH
\] 
as Lie groups. Along the same line of thought, there emerges two immediate realizations of $TTG$; one being $TTG \cong (G\ltimes \G{g}_1)\ltimes (\G{g}_2\ltimes \G{g}_3)$ via \eqref{trTTG}, and the other one being $TTG \cong TG \ltimes T\G{g}_1 \cong (G\ltimes \G{g}_2) \ltimes (\G{g}_1\times \G{g}_3)$ via \eqref{TGH-TG-TH}. The explicit relation between the two was obtained in \cite[(2.58)\,\&\,(2.59)]{EsenSutl17}. Namely,
\begin{equation}\label{12-to-21}
(G\ltimes \G{g}_1)\ltimes (\G{g}_2\ltimes \G{g}_3) \to (G\ltimes \G{g}_2) \ltimes (\G{g}_1\times \G{g}_3), \qquad (g,\xi^{(1)},\xi^{(2)},\xi^{(3)}) \mapsto (g,\xi^{(2)},\xi^{(1)},\xi^{(3)}+\ad_{\xi^{(1)}}\xi^{(2)}),
\end{equation}
with the inverse
\begin{equation}\label{21-to-12}
(G\ltimes \G{g}_2) \ltimes (\G{g}_1\times \G{g}_3) \to (G\ltimes \G{g}_1)\ltimes (\G{g}_2\ltimes \G{g}_3), \qquad (g,\xi^{(2)},\xi^{(1)},\xi^{(3)}) \mapsto (g,\xi^{(1)},\xi^{(2)},\xi^{(3)}-\ad_{\xi^{(1)}}\xi^{(2)}).
\end{equation}
\end{remark}

However, in order to recover the 2nd order Euler-Lagrange equations on $T^2G$ from those on $TTG$, we shall now observe another double cross product decomposition of the iterated tangent group.

\begin{proposition}\label{prop-TTG-g-T2G}
Given any Lie group $G$, with its Lie algebra $\G{g}$ considered as an abelian Lie group; $(\G{g},T^2G)$ is a matched pair of Lie groups. Moreover,
\begin{equation}\label{TTG-g-T2G}
TTG \cong \G{g}\bowtie T^2G .
\end{equation}
\end{proposition}

\begin{proof}
Let us begin with the inclusions 
\[
\G{g}\hookrightarrow TTG \hookleftarrow T^2G
\]
given by
\begin{equation}\label{T2G-to-TTG}
T^2G \to TTG, \qquad (g,\xi,\dot{\xi})\mapsto (g,\xi,\xi,\dot{\xi})
\end{equation}
and
\begin{equation}\label{g-to-TTG}
\G{g}\to TTG, \qquad \xi \mapsto (e,\xi,0,0).
\end{equation}
It takes a routine verification to see that the images of the inclusions \eqref{T2G-to-TTG} and \eqref{g-to-TTG} are indeed subgroups. Furthermore, we see at once that the multiplication on $TTG$ induces
\begin{align*}
& \G{g}\times T^2G \to TTG, \\
& \Big(\tilde{\xi};(g,\xi,\dot{\xi}) \Big) \mapsto (e,\tilde{\xi},0,0)(g,\xi,\xi,\dot{\xi}) = (g,\xi+ \Ad_{g^{-1}}\tilde{\xi},\xi ,\dot{\xi}),
\end{align*}
which is an isomorphism. The claim then follows from \cite[Prop. 6.2.15]{Majid-book}. 
\end{proof}

Let us, for the sake of completeness of the exposition, we now calculate the mutual actions of $T^2G$ and $\G{g}$ on each other. Following the idea in \eqref{mutual-actions}, we calculate
\begin{align*}
& (g,\xi,\xi,\dot{\xi}) (e,\tilde{\xi},0,0)= (g,\xi + \tilde{\xi},\xi,\dot{\xi} +\ad_\xi\tilde{\xi}) = \\
& (e,0,\Ad_g\tilde{\xi},0)(g,\xi,\xi,\dot{\xi}+\ad_\xi\tilde{\xi}).
\end{align*}
Therefore,
\begin{align*}
& (g,\xi,\dot{\xi})\rt \tilde{\xi} =  \Ad_g\tilde{\xi}, \\
& (g,\xi,\dot{\xi}) \lt \tilde{\xi} = (g,\xi,\dot{\xi}+\ad_\xi\tilde{\xi}).
\end{align*}

In order to keep other options of reductions of the Euler-Lagrange equations on $TTG$ we shall record below two more realizations (a 2-cocycle extensions, and a semi-direct product) of $TTG$.

\begin{proposition}
The iterated tangent group $TTG$ is a 2-cocycle extension of $G\ltimes(\G{g}_1\times \G{g}_2)$ by $\G{g}_3$, where $\G{g}_1= \G{g}_2=\G{g}_3=\G{g}$. In short, $TTG \cong \big(G \ltimes (\G{g}_1\times \G{g}_2) \big)\ltimes_\phi \G{g}_3$.
\end{proposition}

\begin{proof}
The semi-direct product group $G \ltimes (\G{g}_1\times \G{g}_2)$ is based on the (right) diagonal adjoint action of $G$ on $\G{g}_1\times \G{g}_2$. More precisely,
\[
(\G{g}_1\times \G{g}_2) \times G \to (\G{g}_1\times \G{g}_2), \qquad( {\xi}^{(1)}, {\xi}^{(2)}) \lt g:= (\Ad_{g^{-1}}\xi^{(1)}, \Ad_{g^{-1}}\xi^{(2)}).
\]
On the other hand, the (right) action of the group $G \ltimes (\G{g}_1\times \G{g}_2)$ on $\G{g}_3$ is the one
\[
\G{g}_3 \times (G\ltimes (\G{g}_1 \times\G{g}_2))\to \G{g}_3, \qquad \xi^{(3)} \lt (g,\xi^{(1)},\xi^{(2)}) := \Ad_{g^{-1}}\xi^{(3)}.
\] 
Finally, the mapping
\begin{equation} \label{Coc-TTG}
\phi: \Big(G\ltimes (\G{g}_1 \times\G{g}_2) \Big) \times \Big(G\ltimes (\G{g}_1 \times\G{g}_2) \Big) \to \G{g}_3, \qquad \Big((g,\xi^{(1)},\xi^{(2)}),(\widetilde{g},\tilde{\xi}^{(1)},\tilde{\xi}^{(2)})\Big) \to \big[\Ad_{\widetilde{g}^{-1}}\xi^{(2)},\tilde{\xi}^{(1)}\big]
\end{equation}
is a 2-cocycle in the group cohomology of the semi-direct group $G\ltimes (\G{g}_1 \times\G{g}_2)$, with coefficients in $\G{g}_3$; namely $\phi \in H^2(G\ltimes (\G{g}_1 \times\G{g}_2), \G{g}_3)$. Indeed,
\begin{align*}
& \phi\Big((\widetilde{g},\tilde{\xi}^{(1)},\tilde{\xi}^{(2)}), (\widetilde{\widetilde{g}},\tilde{\tilde{\xi}}^{(1)},\tilde{\tilde{\xi}}^{(2)})\Big) - \phi\Big((g,\xi^{(1)},\xi^{(2)})(\widetilde{g},\tilde{\xi}^{(1)},\tilde{\xi}^{(2)}), (\widetilde{\widetilde{g}},\tilde{\tilde{\xi}}^{(1)},\tilde{\tilde{\xi}}^{(2)})\Big)  + \\
& \phi\Big((g,\xi^{(1)},\xi^{(2)}),(\widetilde{g},\tilde{\xi}^{(1)},\tilde{\xi}^{(2)}) (\widetilde{\widetilde{g}},\tilde{\tilde{\xi}}^{(1)},\tilde{\tilde{\xi}}^{(2)})\Big) -
	\phi\Big((g,\xi^{(1)},\xi^{(2)}),(\widetilde{g},\tilde{\xi}^{(1)},\tilde{\xi}^{(2)})\Big) \lt (\widetilde{\widetilde{g}},\tilde{\tilde{\xi}}^{(1)},\tilde{\tilde{\xi}}^{(2)}) =\\
&  \phi\Big((\widetilde{g},\tilde{\xi}^{(1)},\tilde{\xi}^{(2)}), (\widetilde{\widetilde{g}},\tilde{\tilde{\xi}}^{(1)},\tilde{\tilde{\xi}}^{(2)})\Big) - \phi\Big((g\widetilde{g},\tilde{\xi}^{(1)} + \Ad_{\widetilde{g}^{-1}}\xi^{(1)}, \tilde{\xi}^{(2)} + \Ad_{\widetilde{g}^{-1}}\xi^{(2)}), (\widetilde{\widetilde{g}},\tilde{\tilde{\xi}}^{(1)},\tilde{\tilde{\xi}}^{(2)})\Big) +\\
&  \phi\Big((g,\xi^{(1)},\xi^{(2)}),(\widetilde{g}\widetilde{\widetilde{g}}, \tilde{\tilde{\xi}}^{(1)} + \Ad_{\widetilde{\widetilde{g}}^{-1}} \tilde{\xi}^{(1)}, \tilde{\tilde{\xi}}^{(2)} + \Ad_{\widetilde{\widetilde{g}}^{-1}}\tilde{\xi}^{(2)}) \Big) -\\
&  \hspace{4cm}  \phi\Big((g,\xi^{(1)},\xi^{(2)}),(\widetilde{g},\tilde{\xi}^{(1)},\tilde{\xi}^{(2)})\Big) \lt (\widetilde{\widetilde{g}},\tilde{\tilde{\xi}}^{(1)},\tilde{\tilde{\xi}}^{(2)}) =\\
&\Big[\Ad_{\widetilde{\widetilde{g}}^{-1}}\tilde{\xi}^{(2)}, \tilde{\tilde{\xi}}^{(1)}\Big] - \Big[\Ad_{\widetilde{\widetilde{g}}^{-1}}(\tilde{\xi}^{(2)} + \Ad_{\widetilde{g}^{-1}}\xi^{(2)}), \tilde{\tilde{\xi}}^{(1)}\Big] + \\
& \hspace{4cm} \Big[\Ad_{(\widetilde{g}\widetilde{\widetilde{g})}^{-1}}\xi^{(2)}, \tilde{\tilde{\xi}}^{(1)} +\Ad_{\widetilde{\widetilde{g}}^{-1}} \tilde{\xi}^{(1)}\Big]  - \Ad_{\widetilde{\widetilde{g}}^{-1}} \Big[\Ad_{\widetilde{g}^{-1}}\xi^{(2)},\tilde{\xi}^{(1)}\Big]  = 0.
\end{align*}
\end{proof}

We now conclude the present subsection with a second presentation of the iterated tangent group.

\begin{proposition}
The iterated tangent group $TTG$ is a semi-direct product of $G$ with 
\[
TT_eG := \{(e,\xi^{(1)};\xi^{(2)},\xi^{(3)}) \mid \xi^{(1)},\xi^{(2)},\xi^{(3)} \in \G{g}\} \subseteq TTG 
\]
given by
\[
\left( \xi^{(1)},\xi^{(2)},\xi^{(3)}\right) \left(
\tilde{\xi}^{(1)},\tilde{\xi}^{(2)},\tilde{\xi}^{(3)}
\right)  =\left( \tilde{\xi}^{(1)}+\xi^{(1)},\tilde{\xi}^{(2)}+\xi^{(2)},\tilde{\xi}^{(3)}+\xi^{(3)}+[\xi^{(2)},\tilde{\xi}^{(1)}]\right).
\]
More precisely,
\[
TTG \cong G\ltimes TT_eG.
\]
\end{proposition}

\begin{proof}
Considering the (right diagonal adjoint) action of $G$ on $TT_eG$ given by
\[
TT_eG  \times G\to TT_eG, \qquad (\xi^{(1)},\xi^{(2)},\xi^{(3)}) \lt g := (\Ad_{g^{-1}}\xi^{(1)}, \Ad_{g^{-1}}\xi^{(2)},\Ad_{g^{-1}}\xi^{(3)}),
\]
it becomes a straightforward verification that the group structure on $TTG$ may be recasted as the semi-direct product as claimed. 
\end{proof}

In fact, the subgroup $TT_eG \subseteq TTG$ is itself a 2-cocycle extension built on the right action
\[
\G{g}_3 \times \big(\G{g}_1 \times \G{g}_2\big) \to \G{g}_3, \qquad \xi^{(3)} \lt (\tilde{\xi}^{(1)},\tilde{\xi}^{(2)}) \to  \xi^{(3)},
\] 
of the abelian (cartesian product) group $\G{g}_1\times \G{g}_2$ on another abelian group $\G{g}_3$, and the 2-cocycle $\chi\in H^2(\G{g}_1\times \G{g}_2,\G{g}_3)$ given by
\[
\chi: (\G{g}_1 \times \G{g}_2) \times (\G{g}_1 \times \G{g}_2) \to \G{g}_3, \qquad \chi\big((\xi^{(1)},\xi^{(2)}), (\tilde{\xi}^{(1)},\tilde{\xi}^{(2)})\big) = [\xi^{(2)},\tilde{\xi}^{(1)}].
\]	
In short, $TT_eG \cong (\G{g}_1 \times \G{g}_2) \ltimes_\chi \G{g}_3$.

\section{Lagrangian dynamics} \label{mpcont}

In the present section, wherein the main results of the paper lie, we study the 2nd order Lagrangian dynamics on the 2nd order tangent (double cross product) group. To this end, we begin with a brief review of the 1st order (matched) Euler-Lagrange equations in Subsection \ref{fold-sec}. It is this subsection in which we also present the 1st order (matched) Euler-Lagrange equations on the iterated tangent group without appealing to any variational calculus (and using only the matched pair theory). Then, in Subsection \ref{sold-sec}, we derive the 2nd order Euler-Lagrange equations on the 2nd order tangent group, merely by the reduction of the 1st order Lagrangian dynamics on the iterated tangent bundle to the 2nd order Lagrangian dynamics, under the symmetry of the Lie algebra of the base Lie group. Finally, we present the 2nd order Euler-Lagrange equations, in Subsection \ref{msoELe}, on the 2nd order tangent group of a double cross product group.

\subsection{First order Lagrangian dynamics}~\label{fold-sec}

In order to arrive at the equation of motion generated by the Lagrangian function(al) $\G{L}:TG\cong G\ltimes \G{g}\to \B{R}$, $\G{L}=\G{L}( g,\xi)$, we compute the variation of the action integral
\begin{equation}
\delta\int_{a}^{b}\G{L}\left( g,\xi\right) dt=\int_{a}^{b}\left(
\left\langle \frac{\delta\G{L}}{\delta g},\delta g\right\rangle
_{g}+\left\langle \frac{\delta\G{L}}{\delta\xi},\delta\xi\right\rangle
_{e}\right) dt.  \label{act}
\end{equation}
Applying the Hamilton's principle to the variations of the base (group)
component, and the reduced variational principle
\begin{equation}\label{rvp}
\delta\xi=\dot{\eta}+\left[ \xi,\eta\right]   
\end{equation}
to the fiber (Lie algebra) component, the (trivialized) Euler-Lagrange equation is computed to be
\begin{equation}\label{preeulerlagrange}
\frac{d}{dt}\frac{\delta\G{L}}{\delta\xi}=T_{e}^{\ast}L_{g}\frac{\delta
	\G{L}}{\delta g}-\ad_{\xi}^{\ast}\frac{\delta\G{L}}{\delta\xi}.
\end{equation}
On the other hand, in view of the (free and proper) left action of $G$ on $TG$, the reduced Euler-Poincar\'e equations generated by the reduced Lagrangian $\G{L}:\G{g}\to \B{R}$, $\G{L}=\G{L}(\xi)$ on the quotient space $TG/G \cong \G{g}$, as (trivial vector) bundles on a point, turn out to be
\begin{equation}\label{EPEq}
\frac{d}{dt}\frac{\delta \G{L}}{\delta\xi}=-\ad_{\xi}^{\ast}\frac{\delta \G{L}}{%
	\delta\xi}.   
\end{equation}
This procedure is called the Euler-Poincar\'{e} reduction. 

For a thorough discussion, we refer the reader to \cite{bou2009hamilton, colombo2011geometry, colombo2013optimal, engo2003partitioned, esen2015tulczyjew}. For the reduced
variational principle we refer to \cite{cendra2003variational, esen2014tulczyjew,MarsdenRatiu-book}, and for
the Lagrangian dynamics on semidirect products we refer to \cite{cendra1998lagrangian,
	n2001lagrangian, holm1998euler, MarsRatiWein84}. 

Next, we shall recall from \cite[Subsect. 3.2]{EsenSutl17} the expressions of \eqref{rvp} and \eqref{EPEq} in the presence of a double cross product group $G\bowtie H$. In this case, given a Lagrangian $\G{L}:(G\bowtie H)\ltimes (\G{g}\bowtie\G{h})$, $\G{L}=\G{L}(g,h,\xi,\eta)$, the Euler-Lagrange equations associated to it were computed to be 
\begin{align}\label{mEL}
\begin{split}
& \frac{d}{dt}\frac{\delta\mathfrak{L}}{\delta\xi}    =T_{e}^{\ast}%
L_{g}\left(  \frac{\delta\mathfrak{L}}{\delta g}\right)  \overset{\ast
}{\lt}h+T_{e}^{\ast}\sigma_{h}\left(  \frac{\delta
	\mathfrak{L}}{\delta h}\right)  -\ad_{\xi}^{\ast}\frac{\delta\mathfrak{L}%
}{\delta\xi}+\frac{\delta\mathfrak{L}}{\delta\xi}\overset{\ast}%
{\lt}\eta+\mathfrak{a}_{\eta}^{\ast}\frac{\delta\mathfrak{L}%
}{\delta\eta},\\
& \frac{d}{dt}\frac{\delta\mathfrak{L}}{\delta\eta}    =T_{e}^{\ast}%
L_{h}\left(  \frac{\delta\mathfrak{L}}{\delta h}\right)  -\ad_{\eta}^{\ast
}\frac{\delta\mathfrak{L}}{\delta\eta}-\xi\overset{\ast}{\rt
}\frac{\delta\mathfrak{L}}{\delta\eta}-\mathfrak{b}_{\xi}^{\ast}\frac
{\delta\mathfrak{L}}{\delta\xi},
\end{split}
\end{align}
where 
\[
\G{g}^\ast\times H \to \G{g}^\ast, \qquad (\mu,h) \mapsto \mu\overset{\ast}{\lt} h,
\]
given by 
\[
\Big\langle \mu\overset{\ast}{\lt} h, \xi \Big\rangle = \Big\langle \mu, h\rt \xi \Big\rangle
\]
is the right action of $H$ on $\G{g}^\ast$, and
\[
G\times \G{h}^\ast \to \G{h}^\ast, \qquad (g,\nu) \mapsto g\overset{\ast}{\rt} \nu,
\]
given by 
\[
\Big\langle g\overset{\ast}{\rt} \nu, \eta \Big\rangle = \Big\langle \nu, \eta\lt g \Big\rangle
\]
is the left action of $G$ on $\G{h}^\ast$. Moreover, $\G{a}_\eta:\G{h}^\ast\to \G{g}^\ast$ is the transpose of the mapping
\[
\G{a}_\eta:\G{g}\to\G{h}, \qquad \xi\mapsto \G{a}_\eta(\xi):=\eta\lt \xi,
\]
and similarly $\G{b}_\xi:\G{g}^\ast\to \G{h}^\ast$ is the transpose of 
\[
\G{b}_\xi:\G{h}\to\G{g}, \qquad \eta\mapsto \G{b}_\xi(\eta):=\eta\rt \xi.
\]
Finally, $\s_h:G\to H$ being the map defined by $\s_h(g):=h\lt g$, $T_e^\ast\s_h:T^\ast_hH \to \G{g}^\ast$ is its coadjoint lift. The symmetry of \eqref{mEL} under the (left) symmetry of $G\bowtie H$, thus, yields the matched Euler-Poincar\'e equations
\begin{align}
\begin{split}\label{mEP}
\frac{d}{dt}\frac{\delta\mathfrak{L}}{\delta\xi}  &  =-\ad_{\xi}^{\ast}%
\frac{\delta\mathfrak{L}}{\delta\xi}+\frac{\delta\mathfrak{L}}{\delta\xi
}\overset{\ast}{\lt}\eta+\mathfrak{a}_{\eta}^{\ast}\frac
{\delta\mathfrak{L}}{\delta\eta},\\
\frac{d}{dt}\frac{\delta\mathfrak{L}}{\delta\eta}  &  =-\ad_{\eta}^{\ast}%
\frac{\delta\mathfrak{L}}{\delta\eta}-\xi\overset{\ast}{\rt
}\frac{\delta\mathfrak{L}}{\delta\eta}-\mathfrak{b}_{\xi}^{\ast}\frac
{\delta\mathfrak{L}}{\delta\xi}.
\end{split}
\end{align}

As an application of \eqref{trTG} and \eqref{mEL}, we readily have the Euler-Lagrange equations on 
\[
TTG\cong T(G\ltimes \G{g}_1) \cong (G\ltimes \G{g}_1) \ltimes (\G{g}_2\ltimes \G{g}_3).
\]
Indeed, considering $h:=\xi^{(1)}$, $\xi:=\xi^{(2)}$, and $\eta:=\xi^{(3)}$, and keeping in mind that the left $\G{g}_1$-action on $G$ and the left $\G{g}_3$-action on $\G{g}_2$ are in this case trivial, and that $\G{g}_3$ is a trivial Lie algebra (being the Lie algebra of the abelian Lie group $\G{g}_1$), we arrive at the Euler-Lagrange equations
\begin{align}\label{UnEPTTG}
\begin{split}
& \frac{d}{dt}\frac{\delta\mathfrak{L}}{\delta\xi^{(2)}}    =T_{e}^{\ast}%
L_{g}\left(  \frac{\delta\mathfrak{L}}{\delta g}\right) - \ad_{\xi^{(1)}}^{\ast}\frac{\delta\mathfrak{L}%
}{\delta\xi^{(1)}} -\ad_{\xi^{(2)}}^{\ast}\frac{\delta\mathfrak{L}%
}{\delta\xi^{(2)}}-\ad_{\xi^{(3)}}^{\ast}\frac{\delta\mathfrak{L}%
}{\delta\xi^{(3)}},\\
& \frac{d}{dt}\frac{\delta\mathfrak{L}}{\delta\xi^{(3)}}    =  \frac{\delta\mathfrak{L}}{\delta \xi^{(1)}}  -\ad_{\xi^{(2)}}^{\ast
}\frac{\delta\mathfrak{L}}{\delta\xi^{(3)}}
\end{split}
\end{align}
generated by a Lagrangian $\G{L}=\G{L}(g,\xi^{(1)},\xi^{(2)},\xi^{(3)})$. Let us note that the equations \eqref{UnEPTTG} may also be expressed as 
\begin{align}\label{UnEPTTG-II}
\begin{split}
& \left(\frac{d}{dt} + \ad_{\xi^{(2)}}^\ast\right)\frac{\delta\mathfrak{L}}{\delta\xi^{(2)}}    =T_{e}^{\ast}%
L_{g}\left(  \frac{\delta\mathfrak{L}}{\delta g}\right) - \ad_{\xi^{(1)}}^{\ast}\frac{\delta\mathfrak{L}%
}{\delta\xi^{(1)}} -\ad_{\xi^{(3)}}^{\ast}\frac{\delta\mathfrak{L}%
}{\delta\xi^{(3)}},\\
& \left(\frac{d}{dt} + \ad^\ast_{\xi^{(2)}}\right)\frac{\delta\mathfrak{L}}{\delta\xi^{(3)}}    =  \frac{\delta\mathfrak{L}}{\delta \xi^{(1)}}  
\end{split}
\end{align}

As a result, after the reduction under the (left) action of $G\ltimes \G{g}_1$, the Euler-Poincar\'e equations generated by the reduced Lagrangian $\G{L}:\G{g}_2\ltimes \G{g}_3 \to \B{R}$, $\G{L}=\G{L}(\xi^{(2)},\xi^{(3)})$ appear to be
\begin{align}\label{EPgg}
\begin{split}
& \left(\frac{d}{dt} + \ad_{\xi^{(2)}}^\ast\right)\frac{\delta\mathfrak{L}}{\delta\xi^{(2)}}    =  -\ad_{\xi^{(3)}}^{\ast}\frac{\delta\mathfrak{L}%
}{\delta\xi^{(3)}},\\
& \left(\frac{d}{dt} + \ad^\ast_{\xi^{(2)}}\right)\frac{\delta\mathfrak{L}}{\delta\xi^{(3)}}    =   0.
\end{split}
\end{align}

\subsection{Second order Lagrangian dynamics}~ \label{sold-sec}

We shall now observe that it is indeed possible to capture the (2nd order) Euler-Lagrange equations on $T^2G \cong (G \ltimes \G{g}_2)\ltimes_\vp \G{g}_3$, and then the (2nd order) Euler-Poincar\'e equations on $G\backslash T^2G \cong \G{g}_2\ltimes \G{g}_3$, from the Euler-Lagrange equations \eqref{UnEPTTG-II} on $TTG \cong (G\ltimes \G{g}_1)\ltimes (\G{g}_2\ltimes \G{g}_3)$, without appealing to any variational calculus, for which we refer to \cite{abrunheiro2011cubic, colombo2011geometry, colombo2013optimal, colombo2014unified,gay2012invariant,gay2011higher}.

Indeed, it follows from Proposition \ref{prop-TTG-g-T2G} that the Euler-Lagrange equations on $TTG \cong G\backslash TTG$, associated to the reduced Lagrangian $\mathfrak{L}:T^2G\to \B{R}$, $\mathfrak{L} = \G{L}(g,\xi,\dot{\xi})$ may be obtained by setting $\xi^{(2)}:=\xi$, and $\xi^{(3)}:=\dot{\xi}$. Namely,
\begin{align}\label{UnEPTTG-III}
\begin{split}
& \left(\frac{d}{dt} + \ad_{\xi}^\ast\right)\frac{\delta\mathfrak{L}}{\delta\xi}    =T_{e}^{\ast}%
L_{g}\left(  \frac{\delta\mathfrak{L}}{\delta g}\right) -\ad_{\dot{\xi}}^{\ast}\frac{\delta\mathfrak{L}%
}{\delta\dot{\xi}},\\
& \left(\frac{d}{dt} + \ad^\ast_{\xi}\right)\frac{\delta\mathfrak{L}}{\delta\dot{\xi}}    =  0.  
\end{split}
\end{align}
Next, incorporating the second equation into the first one along the lines of
\begin{align}\label{idi}
\begin{split}
\ad_{\dot{\xi}}^{\ast}\frac{\delta \G{L}}{\delta\dot{\xi}} & =\frac{d}{dt}\left(
\ad_{\xi}^{\ast}\frac{\delta \G{L}}{\delta\dot{\xi}}\right) -\ad_{\xi}^{\ast}\frac{d}{dt}\left( \frac{\delta \G{L}}{\delta\dot{\xi}}\right)  \\
& =-\frac{d^{2}}{dt^{2}}\left( \frac{\delta \G{L}}{\delta\dot{\xi}}\right)
-\ad_{\xi}^{\ast}\frac{d}{dt}\left( \frac{\delta \G{L}}{\delta\dot{\xi}}\right) \\
& =-\left( \frac{d}{dt}+\ad_{\xi}^{\ast}\right) \frac{d}{dt}\left( \frac{\delta \G{L}}{\delta\dot{\xi}}\right),   
\end{split}
\end{align}
where we used the second equation of \eqref{UnEPTTG-II} in the second equation, we arrive at the second order Euler-Lagrange equations 
\begin{equation}\label{soEL}
\left(\frac{d}{dt} + \ad_{\xi}^\ast\right)\left(\frac{\delta\mathfrak{L}}{\delta\xi} - \frac{d}{dt} \frac{\delta\mathfrak{L}}{\delta\dot{\xi}} \right)  =T_{e}^{\ast}%
L_{g}\left(  \frac{\delta\mathfrak{L}}{\delta g}\right) 
\end{equation}
on $T^2G \cong (G\ltimes \G{g})\ltimes_\vp \G{g}$. Accordingly, one final reduction with respect to the (left) $G$-action renders the 2nd order Euler-Poincar\'e equations
\begin{equation}\label{soep}
\left(\frac{d}{dt} + \ad_{\xi}^\ast\right)\left(\frac{\delta\mathfrak{L}}{\delta\xi} - \frac{d}{dt} \frac{\delta\mathfrak{L}}{\delta\dot{\xi}} \right)  =0
\end{equation}
associated to the reduced Lagrangian $\mathfrak{L}:G\backslash T^2G \cong \G{g}\ltimes \G{g}\to \B{R}$, $\mathfrak{L} = \G{L}(\xi,\dot{\xi})$.

Let us summarize our discussions/arguments in the following diagram.
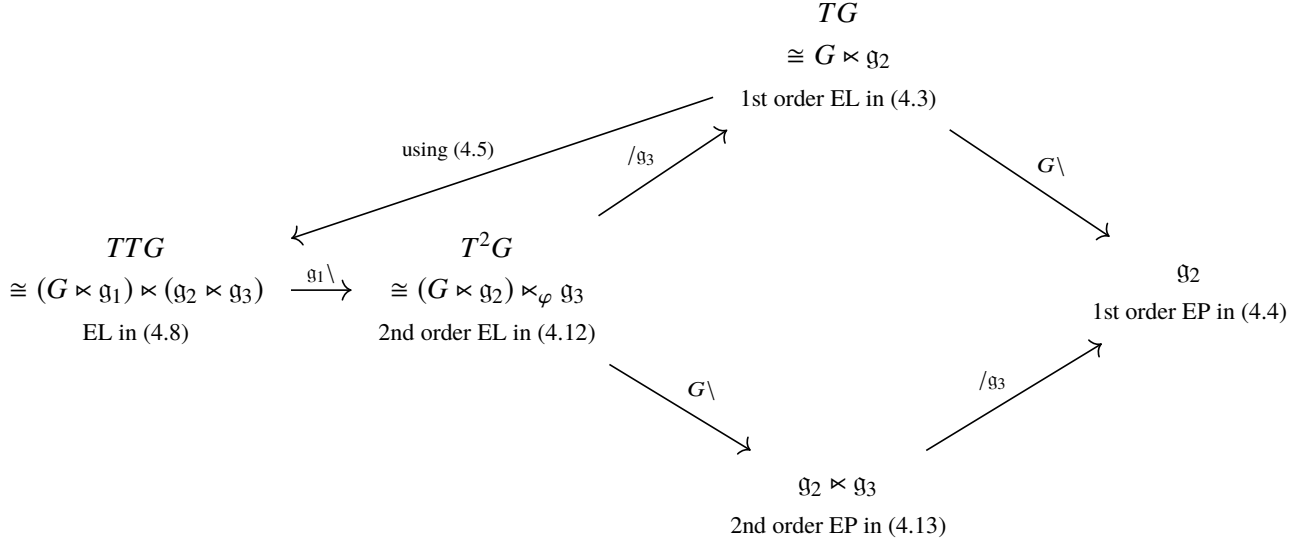
\[
\begin{tikzcd}[column sep={12em,between origins},row sep=3em] 
& & {\begin{array}{c}TG \\ {\footnotesize \cong G\ltimes \G{g}_2} \\ {\footnotesize \text{1st order EL in (\ref{preeulerlagrange})}} \end{array}} \arrow[rd, "{G\backslash}"]\arrow[dll, "{\text{using \eqref{mEL}}}"'] & &\\
{\begin{array}{c}TTG  \\ {\footnotesize \cong (G\ltimes \G{g}_1)\ltimes (\G{g}_2\ltimes \G{g}_3)} \\ {\footnotesize \text{EL in (\ref{UnEPTTG-II})}}\end{array}}
\arrow[r,"{\G{g}_1\backslash}"] &
{\begin{array}{c} T^2G \\ {\footnotesize \cong (G\ltimes \G{g}_2})\ltimes_\vp \G{g}_3 \\ {\footnotesize \text{2nd order EL in (\ref{soEL})}}\end{array}} \arrow[ru, "{/\G{g}_3}"] \arrow[rd, "{G\backslash}"] 
& &{\begin{array}{c} \G{g}_2 \\ {\footnotesize \text{1st order EP in (\ref{EPEq})}}\end{array}}	 \\ 
& & 	{\begin{array}{c}\G{g}_2\ltimes \G{g}_3 \\\footnotesize{\text{2nd order EP in (\ref{soep})}}\end{array}}\arrow[ru, "{/\G{g}_3}"] &
\end{tikzcd}
\]

\subsection{The 2nd order matched Euler-Lagrange equations}~\label{msoELe} 

We shall now re-investigate the 2nd order Euler-Lagrange equations \eqref{soEL}, and the 2nd order Euler-Poincar\'e equations \eqref{soep} in the presence of a matched pair of groups $(G,H)$.

\begin{proposition} \label{Prop-mEL-2nd}
The 2nd order Euler-Lagrange equations on $T^2(G\bowtie H)$, generated by a Lagrangian function $\G{L}:T^2(G\bowtie H)\to \B{R}$, $\G{L}=\G{L}(g,h,\xi,\eta,\dot{\xi},\dot{\eta})$, may be given by
	\begin{equation} \label{mEL-2nd}
	\begin{split}
	\left(\frac{d}{dt}+\ad_{\xi}^* \right)\left(D_{\xi}\G{L}\right)-\left(D_{\xi}\G{L}\right)\overset{\ast }{\vartriangleleft} \eta-\mathfrak{a}_{\eta}^*\left(D_{\eta}\G{L} \right)&=T_{e}^*L_g \left( \frac{\delta \G{L}}{\delta g} \right)\overset{\ast }{\vartriangleleft}h +T_{e}^*\sigma_h\left(\frac{\delta \G{L}}{\delta h} \right), 
	\\
	\left(\frac{d}{dt}+\ad_{\eta}^* \right)\left(D_{\eta}\G{L}  \right) + \xi \overset{\ast }{\vartriangleright}\left(D_{\eta}\G{L}   \right) + \mathfrak{b}_{\xi}^*\left( D_{\xi}\G{L}   \right)&=T_{e}^*L_h \left( \frac{\delta \G{L}}{\delta h} \right)
	\end{split}
	\end{equation}
where we use the abbreviations
\[
D_{\xi}\G{L}:=\frac{\delta \G{L}}{\delta \xi}-\frac{d}{dt}\frac{\delta \G{L}}{\delta \dot{\xi}}, \qquad D_{\eta}\G{L}:=\frac{\delta \G{L}}{\delta \eta}-\frac{d}{dt}\frac{\delta \G{L}}{\delta \dot{\eta}}.
\]
\end{proposition}

\begin{proof}
Setting $g\to (g,h)$, $\xi \to (\xi,\eta)$, and $\dot{\xi}\to (\dot{\xi},\dot{\eta})$ in \eqref{soEL}, we readily get
\[
\left( \frac{d}{dt}+\ad_{(\xi,\eta)}^{\ast}\right) \bigg(\left( \frac{\delta L}{\delta \xi}, \frac{\delta L}{\delta \eta}\right)
	-\frac{d}{dt}\left( \frac{\delta L}{\delta\dot{\xi}}, \frac{\delta L}{\delta\dot{\eta}}\right) \bigg)
	=T^{\ast}L_{(g,h)}\left(\frac{\delta L}{\delta g},\frac{\delta L}{\delta h}\right).
\]
The claim, then, follows from the expression \cite[(3.6)]{EsenSutl17} of the coadjoint lift on a double cross product Lie group, and \cite[(3.6)]{EsenSutl17} of the coadjoint action on a double cross sum Lie algebra.
\end{proof}

The reduction by the (left) action of the double cross product group $G\bowtie H$ is now immediate, and yields the 2nd order matched Euler-Poincar\'e equations on the double cross sum Lie algebra 
\begin{equation}\label{2nd-order-mEP}
(\G{g}\bowtie \G{h})\ltimes (\dot{\G{g}}\times \dot{\G{h}}) \cong (\G{g}\ltimes \dot{\G{g}}) \bowtie (\G{h}\ltimes \dot{\G{h}}),
\end{equation}
which is induced by \eqref{21-to-12}.
 
\begin{corollary} \label{Prop-mEP-2nd}
The 2nd order matched Euler-Poincar\'{e} equations, generated by a Lagrangian function $\G{L}:(\G{g}\bowtie \G{h})\ltimes (\dot{\G{g}}\ltimes \dot{\G{h}}) \to \B{R}$, $\G{L}=\G{L}(\xi,\eta,\dot{\xi},\dot{\eta})$ are given by 
\begin{align} \label{mEP-2nd}
\begin{split}
&\left(\frac{d}{dt}+\ad_{\xi}^* \right)\left(D_\xi\G{L}  \right)-\left( D_\xi\G{L}\right)\overset{\ast }{\vartriangleleft} \eta-\mathfrak{a}_{\eta}^*\left(D_\eta\G{L}\right)=0,\\
&\left(\frac{d}{dt}+\ad_{\eta}^* \right)\left(D_\eta\G{L} \right) + \xi \overset{\ast }{\vartriangleright}\left(D_\eta\G{L} \right) + \mathfrak{b}_{\xi}^*\left( D_\xi\G{L}   \right)=0,
\end{split}
\end{align}
where, once again, we use the abbreviations
\[
D_{\xi}\G{L}:=\frac{\delta \G{L}}{\delta \xi}-\frac{d}{dt}\frac{\delta \G{L}}{\delta \dot{\xi}}, \qquad D_{\eta}\G{L}:=\frac{\delta \G{L}}{\delta \eta}-\frac{d}{dt}\frac{\delta \G{L}}{\delta \dot{\eta}}.
\]
\end{corollary}

Particular instances of the system \eqref{mEP-2nd} are of wide interest in connection with the semi-direct product theory. More precisely, if the (right) action of $G$ on $H$ is assumed to be trivial, then $G\bowtie H = G\rtimes H$, and the 2nd order Euler-Poincar\'e equations on the semi-direct sum Lie algebra
\[
(\G{g}\rtimes \G{h})\ltimes (\dot{\G{g}}\times \dot{\G{h}}) \cong (\G{g}\ltimes \dot{\G{g}})\rtimes (\G{h}\ltimes \dot{\G{h}})
\]
follows at once from \eqref{mEP-2nd} as
\begin{equation} \label{mEP-2nd-sd-1}
\begin{split}
\left(\frac{d}{dt}+\ad_{\xi}^* \right)\left(D_\xi \G{L} \right)-\left( D_\xi \G{L}\right)\overset{\ast }{\vartriangleleft} \eta=0,
\\
\left(\frac{d}{dt}+\ad_{\eta}^* \right)\left(D_\eta \G{L}\right)+ \mathfrak{b}_{\xi}^*\left(  D_\xi \G{L}  \right)=0.
\end{split}
\end{equation}
If, on the other hand, the (left) action of $H$ on $G$ is trivial, then $G\bowtie H = G \ltimes H$, and \eqref{mEP-2nd} reduces this time to the 2nd order Euler-Poincar\'{e} equations on the semi-direct sum Lie algebra 
\[
(\G{g}\ltimes \G{h})\ltimes (\dot{\G{g}}\times \dot{\G{h}}) \cong (\G{g}\ltimes \dot{\G{g}})\ltimes (\G{h}\ltimes \dot{\G{h}})
\]
as
\begin{equation} \label{mEP-2nd-sd-2}
\begin{split}
\left(\frac{d}{dt}+\ad_{\xi}^* \right)\left(D_\xi\G{L}  \right)-\mathfrak{a}_{\eta}^*\left(D_\eta\G{L}\right)=0,
\\
\left(\frac{d}{dt}+\ad_{\eta}^* \right)\left(D_\eta\G{L} \right) + \xi \overset{\ast }{\vartriangleright}\left(D_\eta\G{L} \right)=0.
\end{split}
\end{equation}
If both of the actions are trivial then we arrive at the 2nd order dynamics on $\G{g}\times \dot{\G{g}}$. Namely,
\begin{equation} \label{mEP-2nd-d}
\begin{split}
\left(\frac{d}{dt}+\ad_{\xi}^* \right)\left(D_\xi\G{L}  \right)=0,
\\
\left(\frac{d}{dt}+\ad_{\eta}^* \right)\left(D_\eta\G{L}   \right)=0.
\end{split}
\end{equation}
The following diagram summarizes well the relations between the 1st order and the 2nd order Euler-Lagrange and Euler-Poincar\'e equations derived in Subsection \ref{fold-sec} and Subsection \ref{sold-sec}.
\begin{equation}  \label{diagramm-2}
\begin{tikzcd}[row sep=scriptsize, column sep=scriptsize]
& {\begin{array}{c}T^2G  \\ {\rm \,EL\, in\,} \eqref{soEL} \end{array}}\arrow[dl,"G\backslash"']  \arrow[dd,"/\dot{\G{g}}"',"\text{via } \eqref{T2G-cocycle-ext}"] & {\begin{array}{c}T^2(G\bowtie H)  \\ {\text{matched EL in }} \eqref{mEL-2nd} \end{array}} \arrow[dr,"(G\bowtie H)\backslash"] \arrow[dd,"/(\dot{\G{g}}\times \dot{\G{h}})"',"\text{via } \eqref{T2G-cocycle-ext}"] \arrow[l,"/T^2H"',"\text{via } \eqref{T2GH-matched}"] \\
{\begin{array}{c}\mathfrak{g}\ltimes \dot{\mathfrak{g}}\\ {\rm  \,EP\, in\,} \eqref{soep} \end{array}}  \arrow[dd,"/\dot{\G{g}}"'] & &&{\begin{array}{c}(\mathfrak{g}\bowtie \G{h})\ltimes (\dot{\mathfrak{g}}\times \dot{\mathfrak{h}}) \\ {\text{matched EP in }} \eqref{mEP-2nd} \end{array}} \arrow[lll,"/(\mathfrak{h}\times \dot{\mathfrak{h}})"',"\text{via }\eqref{2nd-order-mEP}"] \arrow[dd,"/(\dot{\G{g}}\times \dot{\G{h}})"]\\
&{\begin{array}{c}TG\\ {\rm EL\, in\,} \eqref{preeulerlagrange} \end{array}} \arrow[dl,"G\backslash"']  &  {\begin{array}{c}T(G\bowtie H)\\ {\text{matched EL in }} \eqref{mEL} \end{array}} \arrow[l, "/TH"',"\text{via } \eqref{TGH-TG-TH}"] \arrow[dr,"(G\bowtie H)\backslash"] \\
{\begin{array}{c}\mathfrak{g}\\ {\rm \,EP\, in\,} \eqref{EPEq} \end{array}}  & && {\begin{array}{c}\mathfrak{g}\bowtie \mathfrak{h}\\ {\text{matched EP in }} \eqref{mEP} \end{array}} \arrow[lll,"/\G{h}"'] \\
\end{tikzcd}
\end{equation}

\section{Illustrations}\label{sect-illustrations}

\subsection{Riemannian 2-splines}\label{subsect-2-splines}~

A class of Lie groups which may be (matched) paired by themselves, via the adjoint action, is given by the nilpotent groups of class 2; \cite{esen2017dinamic,Majid-book}.

In this subsection, we shall present the 2nd order matched Euler-Poincaré equations of the Riemannian 2-splines, see for instance \cite{gay2012invariant,gay2011higher}, on the 2nd order tangent group of a double cross product $G\bowtie G$, where $G$ is a nilpotent Lie group of class 2.

Let us note that in case $G$ is a nilpotent Lie group of class 2, then the Riemannian metric $\gamma$ on the Riemannian 2-splines on $G$ turns into a bi-invariant Riemannian metric \cite{gay2012invariant}. Accordingly, $G$ being a nilpotent group of class 2, let us fix a bi-invariant Riemannian metric $\bar{\gamma}$ on the double cross product group $G\bowtie G$, with the corresponding squared norm $||(\cdot,\cdot)||^2_{\G{g}\bowtie \G{g}}$ on the double cross sum Lie algebra $\G{g}\bowtie \G{g}$ of the group $G\bowtie G$.

Along the lines of \cite[Prop. 3.4]{gay2012invariant}, we shall consider the Lagrangian 
\begin{equation}
\bar{\G{L}}:T^{2}(G\bowtie G)\to \B{R}, \qquad \bar{\G{L}}\left((g,\tilde{g}),(\xi,\tilde{\xi}),(\dot{\xi},\dot{\tilde{\xi}})\right)=\frac{1}{2}\left|\left|\frac{D}{Dt}(\xi,\tilde{\xi})\right|\right|_{(g,\tilde{g})}^2,
\end{equation}
where $D/Dt$ stands for the covariant derivative with respect to time, which induces the reduced Lagrangian 
\[
\bar{\ell}:(\G{g}\bowtie \G{g})\ltimes (\G{g}\times \G{g})\to \B{R}, \qquad \bar{\ell}((\xi,\tilde{\xi}),(\dot{\xi},\dot{\tilde{\xi}}))=\frac{1}{2}\left|\left| (\dot{\xi},\dot{\tilde{\xi}})  \right|\right|_{\G{g}\bowtie\G{g}}^2.
\]
Accordingly, the 2nd order Euler-Poincaré equations \cite[(3.21)]{gay2012invariant} becomes
\begin{align}
\left(\frac{d}{dt}\pm\ad_{\xi}^*\right) \ddot{\xi}^\flat \mp\ad^*_{\tilde{\xi}}(\ddot{\xi}^\flat+\ddot{\tilde{\xi}}^\flat)&=0 \qquad \text{ or }\qquad \dddot{\xi}\mp \left[\xi,\ddot{\xi}\right]\pm \left[\tilde{\xi},\ddot{\xi}+\ddot{\tilde{\xi}}\right]=0, \label{Nilpotent-1}\\
\left(\frac{d}{dt}\pm\ad_{\tilde{\xi}}^*\right) \ddot{\tilde{\xi}}^\flat \pm\ad^*_{\xi}(\ddot{\xi}^\flat+\ddot{\tilde{\xi}}^\flat)&=0 \qquad \text{ or }\qquad \dddot{\tilde{\xi}}\mp \left[\tilde{\xi},\ddot{\tilde{\xi}}\right]\mp \left[\xi,\ddot{\xi}+\ddot{\tilde{\xi}}\right]=0, \label{Nilpotent-2}
\end{align}
where 
\[
\flat:\G{g}\to \G{g}^*, \qquad \xi \mapsto \xi^\flat
\]
is the (musical) isomorphism given by
\[
\left\langle \xi^\flat,\tilde{\xi} \right\rangle=\gamma(\xi,\tilde{\xi}),
\]
whose inverse 
\[
\#:\G{g}^*\to \G{g},\qquad \mu \mapsto \mu^{\#}
\]
may be given by 
\[
-\ad_{\xi}\eta=(\ad_\xi^*(\eta^\flat))^{\#}.
\] 
Let us finally record from \cite{esen2017dinamic} that in the presence of the double cross sum Lie algebra $\G{g}\bowtie \G{g}$, the coadjoint action may be formulated as 
\begin{equation}\label{coadjoint-action}
\ad^*_{(\xi,\tilde{\xi})}(\xi^\flat,\tilde{\xi}^\flat)=(\ad_{\xi}^* \xi^\flat-\ad^*_{\tilde{\xi}}(\xi^\flat+\tilde{\xi}^\flat), \ad_{\tilde{\xi}}^*\tilde{\xi}^{\flat}+\ad^*_{\xi}(\xi^\flat+\tilde{\xi}^\flat)).
\end{equation}

\subsection{The 2nd order $3D$ systems}~\label{subsect-2nd-order-3D}

In the present subsection, we shall examine the $2$nd order matched Euler-Poincaré equations on $(\mathbb{R}^3\bowtie \mathbb{R}^3)\ltimes (\mathbb{R}^3\times \mathbb{R}^3)$ generated by a reduced Lagrangian function 
\[
\ell:(\mathbb{R}^3\bowtie \mathbb{R}^3)\ltimes (\mathbb{R}^3\times \mathbb{R}^3)\to \B{R},\qquad \ell=\ell(X,Y,\dot{X},\dot{Y}).
\]
 In other words, write the equations \eqref{mEP-2nd} of Corollary (\ref{Prop-mEP-2nd}) on $(\mathbb{R}^3\bowtie \mathbb{R}^3)\ltimes (\mathbb{R}^3\ltimes\mathbb{R}^3)$ in view of the double cross sum decomposition 
\begin{equation}
\G{sl}(2,\mathbb{C})=\G{su}(2)\bowtie \G{K}\cong \mathbb{R}^3\bowtie \B{R}^3_{\mathbf{k}},
\end{equation}
associated to the Iwasawa decomposition $SU(2)\bowtie K$ of $SL(2,\B{C})$, where $\G{su}(2)$ is the algebra of the group $SU(2)$, and $\G{K}$ is the Lie algebra of the half-real form $K$ of $SU(2)$. Referring the reader to \cite{EsenSutl17} for further details of this double cross sum decomposition, let us fix the notation
\begin{equation}
\begin{split}
X,\dot{X}\in & \B{R}^3, \qquad Y,Y_1,Y_2,\dot{Y}\in \B{R}^3_{\mathbf{k}},\qquad\hspace{0.2 cm}  
\Phi\in  {\B{R}^3}^*, \qquad \Psi\in \G{K}^*\cong {\B{R}^3_{\mathbf{k}}}^*, \qquad \mathbf{k}=(0,0,1)\in \B{R}^3.  
\end{split}
\end{equation}
We shall identify the Lie algebra $\G{su}(2)$ with $\B{R}^3$ whose Lie algebra structure is given by the cross product, and its dual space $\G{su}^*(2)$ with $\B{R}^3$ (via the Euclidean dot product). The Lie algebra structure of $\G{K}\cong\B{R}^3_{\mathbf{k}}$, on the other hand, is determined by
\[
 [Y_1,Y_2]=\mathbf{k}\times (Y_1\times Y_2),
\]
where $\times$ stands for the cross product, and $\mathbf{k}$ denotes the unit vector $(0,0,1)$. Let us recall also that, with all these identifications at hand, the coadjoint actions of the Lie algebras $\mathbb{R}^3$ and $\B{R}^3_{\mathbf{k}}$ are given by
\begin{equation}
\ad^*:\B{R}^3\times \B{R}^3\to \B{R}^3, \qquad\hspace{0.2 cm} (X,\Phi)\mapsto \ad^*_X\Phi=X\times  \Phi, \label{ad-*-X} 
\end{equation} 
and
\begin{equation}
\ad^*:\B{R}^3_{\mathbf{k}}\times \B{R}^3_{\mathbf{k}}\to \B{R}^3_{\mathbf{k}}, \qquad\hspace{0.2 cm} (Y,\Psi)\mapsto \ad^*_Y\Psi=(\mathbf{k}\cdot Y)\Psi-(\Psi\cdot Y)\mathbf{k}, \label{ad-*-Y} 
\end{equation} 
while the mutual dual actions are 
\begin{align}
&\overset{\ast }{\lt}:\B{R}^3\times \B{R}^3_{\mathbf{k}}\to \B{R}^3, \qquad (\Phi,Y)\mapsto \Phi\overset{\ast }{\lt} Y=(Y\cdot \mathbf{k})\Phi-(\Phi\cdot \mathbf{k})Y \label{left-trng-*}\\
& \overset{\ast }{\rt}:\B{R}^3\times \B{R}^3_{\mathbf{k}}\to \B{R}^3_{\mathbf{k}}, \qquad (X,\Psi)\mapsto X\overset{\ast }{\rt} \Psi=\Psi\times X,\label{right-trng-*}\\
&\G{b}^*_X:\B{R}^3\to \B{R}^3_{\mathbf{k}}, \qquad \Phi \mapsto \G{b}_X^* \Phi=(\Phi\cdot \mathbf{k})X-(\Phi\cdot X)\mathbf{k},\label{b-X-*}\\
&\G{a}^*_Y:\B{R}^3_{\mathbf{k}}\to \B{R}^3, \qquad \Psi \mapsto \G{a}_Y^* \Psi=Y\times \Psi,\label{a-Y-*}
\end{align}
where we use the Euclidean dot product in \eqref{b-X-*}. Accordingly, the 2nd order (matched) Euler-Poincaré equations \eqref{mEP-2nd} appears to be 
\begin{align}
\left(\frac{d}{dt}-Y\cdot \mathbf{k}\right)D_X\ell+X\times D_X\ell-Y\times D_Y\ell+\left(D_X\ell\cdot\mathbf{k}\right)Y&=0,\label{1st-eq-R-3}\\
\left(\frac{d}{dt}+\mathbf{k}\cdot Y\right)D_Y\ell+D_Y\ell\times X+(D_X\ell\cdot \mathbf{k})X-(D_Y\ell\cdot Y+D_X\ell\cdot X)\mathbf{k}&=0,\label{2nd-eq-R-3} 
\end{align}
 where we use the identities
 \begin{equation}\label{id-R-3}
 \begin{split}
 D_X\ell&=\frac{\delta \ell}{\delta X}-\frac{d}{dt}\frac{\delta\ell}{\delta \dot{X}},\\
 D_Y\ell&=\frac{\delta \ell}{\delta Y}-\frac{d}{dt}\frac{\delta\ell}{\delta \dot{Y}}.
 \end{split}
 \end{equation}
 If, in particular, the left action of $K$ on $SU(2)$ is assumed to be trivial, then the equations of motion  \eqref{1st-eq-R-3} and \eqref{2nd-eq-R-3} turn out to be the 2nd order Euler-Poincaré equations
\begin{align*}
\frac{d}{dt}D_X\ell+X\times D_X\ell-Y\times D_Y\ell&=0,\\
\left(\frac{d}{dt}+\mathbf{k}\cdot Y\right)D_Y\ell+D_Y\ell\times X-(D_Y\ell\cdot Y)\mathbf{k}&=0, 
\end{align*}
on the semidirect sum $(\mathbb{R}^3\ltimes \mathbb{R}^3)\ltimes (\mathbb{R}^3\times \mathbb{R}^3)\cong (\mathbb{R}^3\ltimes \mathbb{R}^3)\ltimes (\mathbb{R}^3\ltimes \mathbb{R}^3)$.

If, on the other hand, the right action of $SU(2)$ on $K$ is assumed to be trivial, then the equations  \eqref{1st-eq-R-3} and \eqref{2nd-eq-R-3} transform into the 2nd order Euler-Poincaré equations
\begin{align*}
\left(\frac{d}{dt}-Y\cdot\mathbf{k}\right)D_X\ell+X\times D_X\ell+\left(D_X\ell\cdot\mathbf{k}\right)Y&=0,\\
\left(\frac{d}{dt}+\mathbf{k}\cdot Y\right)D_Y\ell+(D_X\ell\cdot\mathbf{k})X-(D_Y\ell\cdot Y+D_X\ell\cdot X)\mathbf{k}&=0,
\end{align*}
on  $(\mathbb{R}^3\rtimes \mathbb{R}^3)\ltimes (\mathbb{R}^3\times \mathbb{R}^3)\cong (\mathbb{R}^3\ltimes \mathbb{R}^3)\rtimes (\mathbb{R}^3\ltimes \mathbb{R}^3)$.

If both actions are trivial, then we arrive at the 2nd order Euler-Poincar\'e equations 
\begin{align*}
\frac{d}{dt}D_X\ell+X\times D_X\ell&=0,\\
\left(\frac{d}{dt}+\mathbf{k}\cdot Y\right)D_Y\ell-(D_Y\ell\cdot Y)\mathbf{k}&=0
\end{align*}
on $(\B{R}^3\times \B{R}^3)\ltimes (\B{R}^3\times \B{R}^3)$. Furthermore, for the reduced Lagrangian function 
\[
\ell:(\B{R}^3\times \B{R}^3)\ltimes (\B{R}^3\times \B{R}^3)\to \B{R}, \qquad \ell(X,Y,\dot{X},\dot{Y}):=\frac{1}{2}(X^2+Y^2+\dot{X}^2+\dot{Y}^2),
\]
for which we refer the reader to \cite{starostin2015equilibrium}, the last equations take the particular form of 
\begin{align}
\dddot{X}-(Y\cdot\mathbf{k}) \ddot{X}
+X\times\ddot{X}+(\ddot{X}\cdot X)\mathbf{k}-Y\times \ddot{Y}-\dot{X}+(Y\cdot\mathbf{k})X-(X\cdot \mathbf{k})Y
&=0,\\
\dddot{Y}+(\mathbf{k}\cdot Y)\ddot{Y}+\ddot{Y}\times X-(\ddot{Y}\cdot Y)\mathbf{k}+(\ddot{X}\cdot\mathbf{k})X-(\ddot{X}\cdot X)\mathbf{k}-\dot{Y}-(\mathbf{k}\cdot Y)Y-Y\times X+(Y\cdot Y)\mathbf{k}&=0,
\end{align}
where we used the abbreviations
\[
D_X\ell=X-\ddot{X},\qquad D_Y\ell=Y-\ddot{Y}.
\]

\section{Discussions}

We have studied the 2nd order Lagrangian dynamics on Lie groups. More precisely, we have observed that based on the double cross product decomposition \eqref{TTG-g-T2G} of $TTG$, the 2nd order Euler-Lagrange equations on $T^2G$ may be obtained from the 1st order Euler-Lagrange equations on $TTG$; without appealing to any variational calculus. Furthermore, we have also remarked that the (1st order) Euler-Lagrange equations on $TTG$ may be written directly from those on the tangent group of a double cross product group; simply by regarding $TG$ as the semi-direct product $G\ltimes \G{g}$. Finally, we have noted the 2nd order Euler-Lagrange equations on the 2nd order tangent group of a double cross product group.

As is well-known, a generalization of the classical Lagrangian dynamics is available on the Lie algebroid framework; see, for instance, \cite{de2005lagrangian,grabowska2006geometrical,libermann1996lie,martinez2001lagrangian,Wein96}. Moreover, in the Lie algebroid framework, it is also possible to study the higher order Lagrangian dyanmics as well; \cite{colombo2017second,jozwikowski2018higher,martinez2015higher}. Along the lines of the present paper, it is the matched pair interpretation of the Lie algebroid geometry of the higher order Lagrangian dynamics that we plan on investigating in a separate paper.

\section{Acknowledgement} 
This work is a part of the project "Matched pairs of Lagrangian and Hamiltonian Systems"  supported by T\"UB\.ITAK (the Scientific and Technological Research Council of Turkey) with the project number 117F426 pursued by OE and SS.

\bibliographystyle{plain}
\bibliography{references}{}

\end{document}